\documentclass[12pt,a4paper,twoside]{article}
\usepackage{eurosym}
\usepackage[utf8]{inputenc}
\usepackage[english]{babel}
\usepackage{amsmath}
\usepackage{amsfonts}
\usepackage{amssymb}
\usepackage{amsthm}
\usepackage{graphicx}
\usepackage[top=1.3in,bottom=1.3in,left=1in,right=1in]{geometry}
\usepackage{caption,subcaption}
\usepackage{authblk}
\usepackage{lineno,hyperref}
\usepackage{multirow}
\usepackage[justification=centering]{caption}
\usepackage{epstopdf}
\newtheorem{thm}{Theorem}

\newtheorem{proposition}[thm]{Proposition}

\usepackage{listings}
\usepackage{float}

\begin{document}
\date{}
\title{\textbf{A generalization of Ramos-Louzada distribution: Properties and estimation}}
\author[a]{Hazem Al-Mofleh\thanks{almof1hm@cmich.edu}}
\author[b]{Ahmed Z. Afify}
\affil[a] {\small{Department of Mathematics, Tafila Technical University, Tafila, Jordan.}}
\affil[b] {\small{Department of Statistics, Mathematics and Insurance, Benha University, Egypt.}}

\maketitle

\begin{abstract}
In this paper, a new two-parameter model called generalized Ramos-Louzada (GRL) distribution is proposed. The new model provides more flexibility in modeling data with increasing, decreasing, j shaped and reversed-J shaped hazard rate function. Several statistical and reliability properties of the GRL model are also presented in this paper. The unknown parameters of the GRL distribution are discussed using eight frequentist estimation approaches. These approaches are important to develop a guideline to choose the best method of estimation for the GRL parameters, that would be of great interest to practitioners and applied statisticians. A detailed numerical simulation study is carried out to examine the bias and the mean square error of the proposed estimators. We illustrate the performance of the GRL distribution using two real data sets  from the fields of medicine and geology and both data sets show that the new model is more appropriate as compared to the gamma, Marshall-Olkin exponential, exponentiated exponential, beta exponential, generalized Lindley, Poisson-Lomax, Lindley geometric and Lindley distributions, among others.\\
\textbf{Keywords:} Cram\'{e}r--von Mises Estimation; Maximum Likelihood Estimation; Maximum Product of Spacing Estimation; Right--Tail Anderson-Darling Estimation.\\
\textbf{AMS subject classification:} 62E10, 60K10, 60N05.
\end{abstract}

\section{Introduction}
\label{sec1}

The probability distributions have great importance for modeling data in several areas such as medicine, engineering, and life testing, among others. Ramos and Louzada (2019) recently introduced the one-parameter distribution called Ramos-Louzada (RL) distribution with survival function (SF) given by
\begin{equation}\label{fdpnda}
S(t|\lambda) = \left(\frac{1}{\lambda-1}\right)
\left(\lambda-1+\frac{t}{\lambda} \right)e^{- \frac{t}{\lambda}}, \quad \quad t>0,
\end{equation}
where $\lambda\geq 2$.

The two most common one-parameter distributions are the exponential and Lindley distributions. The important generalizations of the exponential distribution are the Weibull (Weibull, 1951) and exponentiated exponential (Gupta and Kundu, 2001) models. In the case of the Lindley distribution, the power Lindley (Ghitany et al., 2013) and generalized Lindley (Nadarajah et al., 2011) models have play an important role in survival analysis. These two generalizations are obtained by considering a power parameter in the exponential and Lindley distributions.
Ramos and Louzada (2019) showed that (\ref{fdpnda}) outperforms the common exponential and Lindley distributions in many situations.  Therefore, we will propose a new two-parameter extension of the RL distribution by including a power parameter in the baseline model (\ref{fdpnda}). The new proposed model is called a generalized Ramos-Louzada (GRL) distribution.

Let $T$ be a non-negative random variable follows the GRL model, the SF of random variable $T$ is given by
\begin{equation}\label{fdpnd}
S(t|\lambda,\alpha) = \left(\frac{1}{\lambda-1}\right)
\left(\lambda-1+\frac{t^\alpha}{\lambda} \right)e^{- \frac{t^\alpha}{\lambda}},
\end{equation}
where $\lambda(\geq 2)$ and $\alpha (>0)$ are shape parameters.

Some mathematical properties, parameters estimation by eight different methods, simulations and applications are studied and proposed in this paper.

We can summarize the motivations of this proposed model as:
(i) Its cumulative distribution function (CDF) and hazard rate function (HRF) have simple closed forms, hence, it can be utilized to analyze censored data; (ii) It can be represented as a mixture of Weibull distribution and a particular case of the generalized gamma distribution (Stacy, 1962) (See Section~\ref{sec2}); (iii) The GRL distribution exhibits increasing, decreasing, reversed-J shaped and J shaped hazard rates, whereas the RL model exhibits only increasing hazard rate; and (iv) The GRL distribution outperforms many of the well-known distributions namely: gamma, Marshall-Olkin exponential, exponentiated exponential, Beta exponential, generalized Lindley, Poisson-Lomax, Lindley geometric and Lindley distributions, using two real data sets from the fields of medicine and geology.
Furthermore, another important goal of this paper is to show how several frequentist estimators of the GRL parameters perform to choose the best parameter estimation method for the proposed model, which would be a great interest to practitioners and applied statisticians.

The skewness of the GRL distribution varies within the interval (-0.68158, 5.17333), whereas the skewness of the RL distribution can only range in the interval (1.41421, 1.85648) when the parameter $\lambda$ takes values (2, 3.1, 4, 5.5). Furthermore, the spread of the kurtosis of the GRL distribution is much larger ranging, which is from 2.69447 to 52.6597, whereas the spread of the kurtosis of the RL distribution can only varies from 6.00 to 8.04 for the same values shown above for the parameter $\lambda$.

This paper is organized as follows: Section~\ref{sec2} introduces the GRL distribution and its properties such as: quantile function, moments, order statistics and HRF. Section~\ref{sec3} presents the estimators of the GRL unknown parameters based on eight classical estimation methods. Simulation study, to evaluate and compare the behavior of the eight classical estimation methods, is discussed in Section~\ref{sec4}. Section~\ref{sec5} illustrates the relevance of GRL model for two real lifetime data sets.  Section~\ref{sec6} summarizes the present study.

\section{Properties}
\label{sec2}

Let $T$ be a random variable follows the GRL model with SF given in (\ref{fdpnd}), the probability density function (PDF) of the random variable $T$ is given by

\begin{equation}\label{RLPDF}
\quad \quad f(t;\pmb \phi)=\frac{\alpha}{\lambda(\lambda-1)}t^{\alpha-1}\left(\lambda+\frac{t^\alpha}{\lambda} -2\right)e^{-\frac{t^\alpha}{\lambda} }, \quad t>0, \quad \lambda\geq 2, \quad \alpha>0,
\end{equation}

\noindent where $\pmb \phi=(\lambda,\alpha)^{\intercal}$.
Note that, the RL model can be obtained from (\ref{RLPDF}) when $ \alpha=1$.

\subsection{Shapes}

The behavior of the PDF in (\ref{RLPDF}) when $t\rightarrow0$ and $t\rightarrow\infty$ are, respectively, given by

\begin{equation*}
\lim_{t\rightarrow 0}f(t;\pmb \phi)=
\begin{cases}
 \infty, & \text{if }\alpha<1 \\
\dfrac{(\lambda-2)}{\lambda(\lambda-1)}, & \text{if }\alpha=1 \\
 0, & \text{if }\alpha>1
\end{cases},
\end{equation*}

\begin{equation*}
\lim_{t\rightarrow \infty} f(t;\pmb \phi)= 0.
\end{equation*}

In Figure \ref{density}, we present the shapes of the PDF for different values of the parameters $\lambda$ and $\alpha$. The shape of PDF of the GRL model can be right-skewed and reversed-J shaped.

\begin{figure}[H]
\begin{subfigure}{.5\textwidth}
  \centering
  \includegraphics[width=1\linewidth]{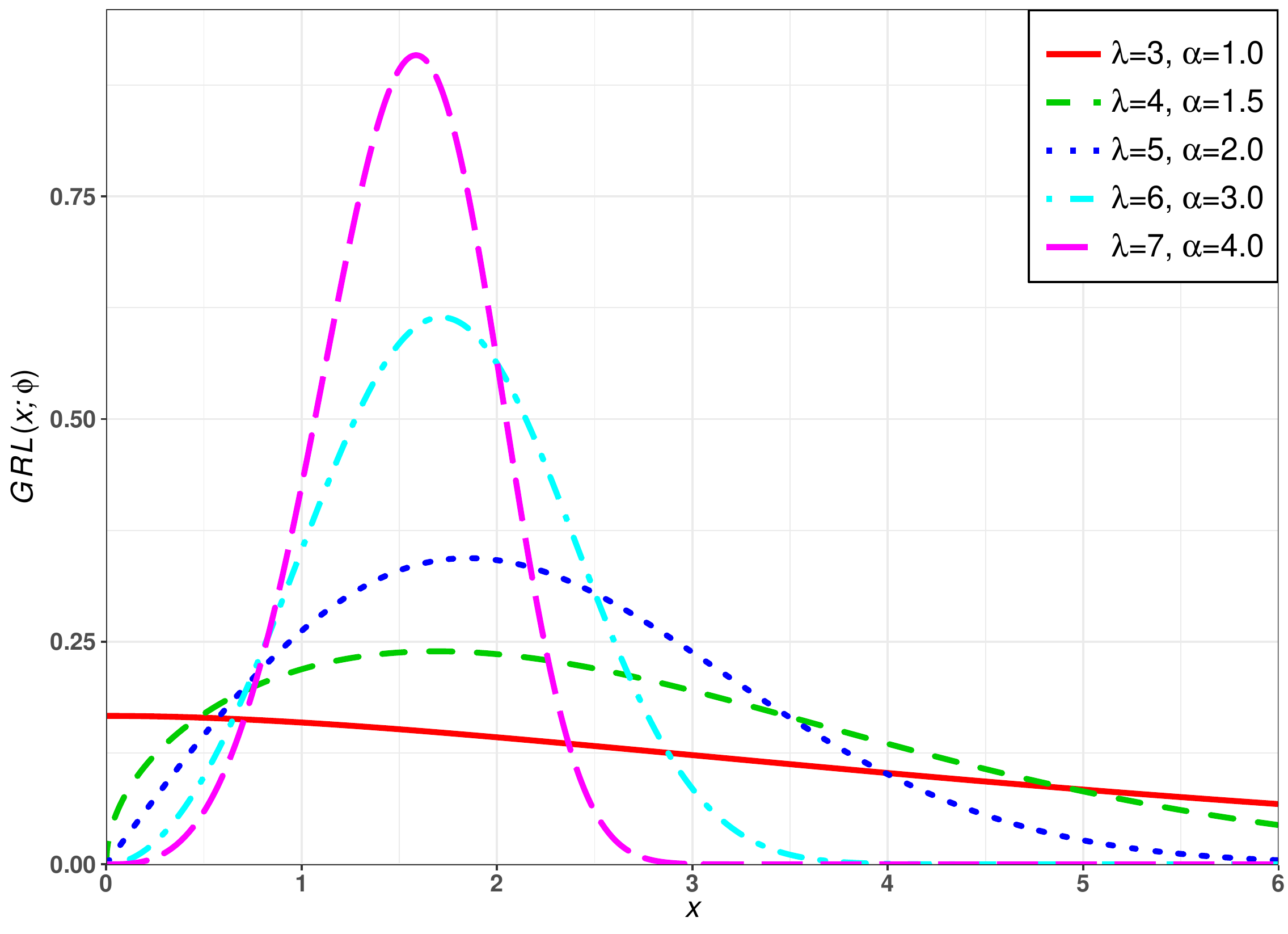}
  \caption*{}
\end{subfigure}
    \begin{subfigure}{.5\textwidth}
  \centering
  \includegraphics[width=1\linewidth]{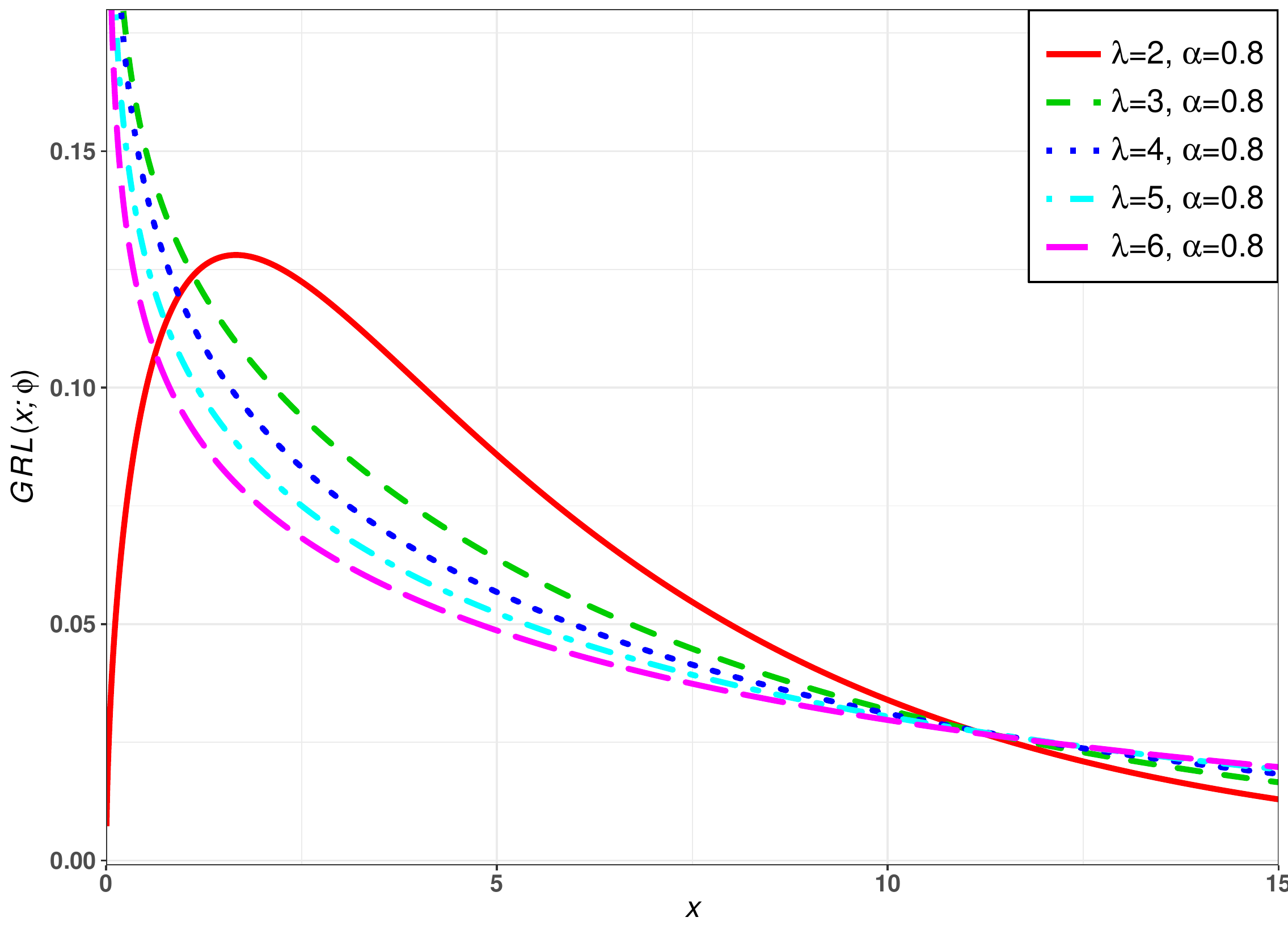}
  \caption*{}
\end{subfigure}
\caption{PDF shapes for the GRL distribution considering different
values of $\lambda $ and $\alpha$.}
  \label{density}
\end{figure}

The CDF of the GRL distribution is given by

\begin{equation}\label{cdf}
F(t;\pmb \phi) = 1-\left(\frac{1}{\lambda-1}\right)
\left(\lambda-1+\frac{t^\alpha}{\lambda} \right)e^{- \frac{t^\alpha}{\lambda}}, \quad t>0, \quad \lambda\geq 2, \quad \alpha>0. \end{equation}

The HRF of $T$ is given by
\begin{equation}\label{fhwl}
h(t;\pmb \phi)=\frac{f(t|\alpha,\lambda)}{S(t|\alpha,\lambda)}=\frac{\alpha t^\alpha}{\lambda}\frac{(\lambda^2+t^\alpha-2\lambda)}{(\lambda^2+t^\alpha-\lambda)} .
\end{equation}

Figure \ref{HRF} displays some possible shapes of HRF of the GRL for some selected values of $\lambda$ and $\alpha$. The shape of HRF can be increasing, decreasing, reversed-J shaped and J shaped hazard rates.

\begin{figure}[!h]
\centering
\includegraphics[scale=0.5]{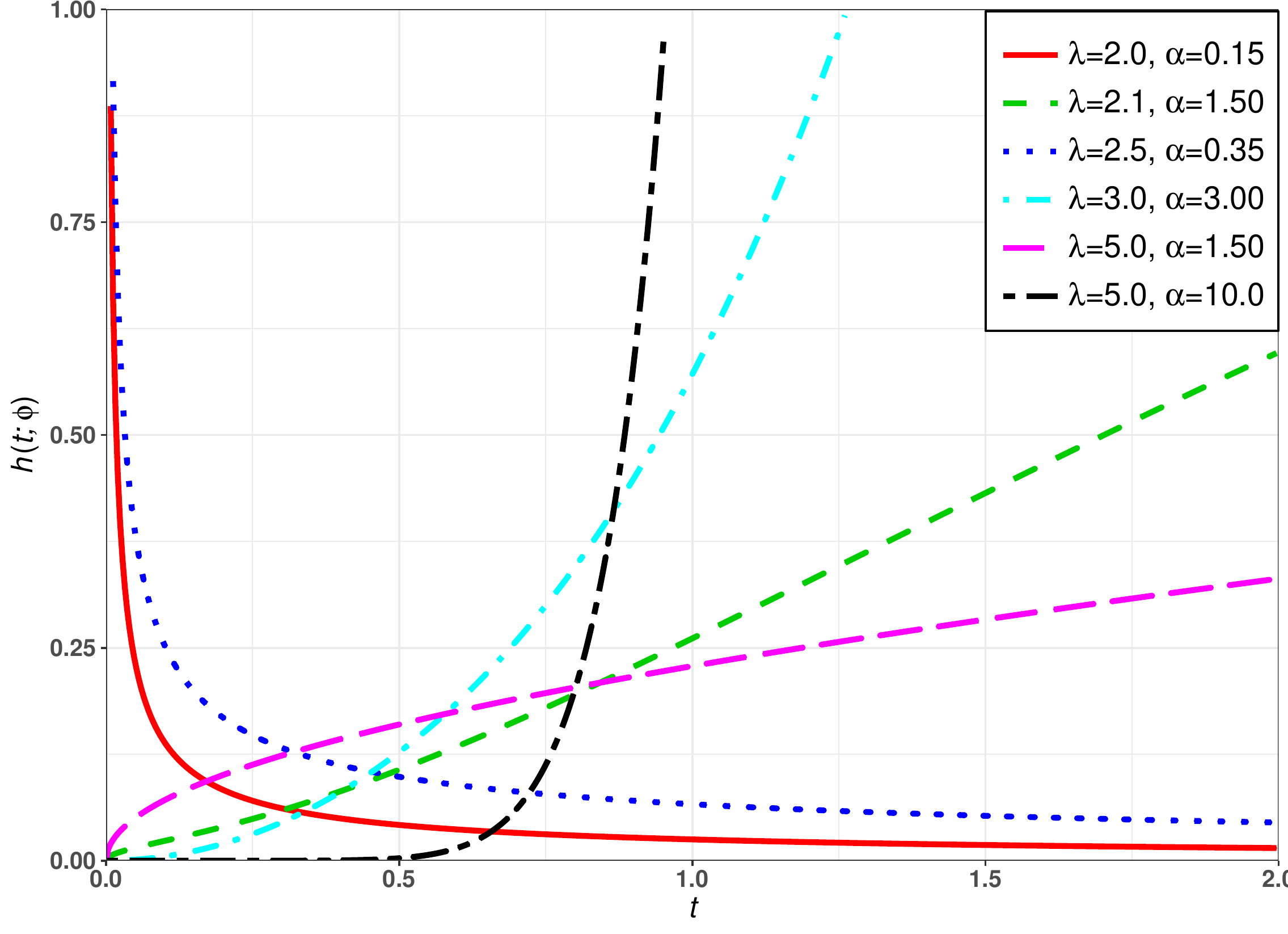}
\caption{Hazard rat function shapes for the GRL considering
different values of $\lambda $ and $\alpha$.}
\label{HRF}
\end{figure}

The GRL distribution can be expressed as a two-component mixture
\begin{equation}\label{eqmix2}
f(t;\pmb \phi)= pf_1(t;\pmb \phi)+(1-p)f_2(t;\pmb \phi)\, ,
\end{equation}
where  $1-p=1/(\lambda-1)$ (or $p=(\lambda-2)/(\lambda-1)$) and
\begin{equation}\label{mixeq1} f_j(t;\pmb \phi)=\frac{\alpha}{\lambda^j} t^{j\alpha-1}e^{- \frac{t^\alpha}{\lambda}} \quad \mbox{for} \quad  j=1,2.\end{equation}

Note that, $f_1(\cdot)$ is a Weibull distribution and $f_2(\cdot)$ is a particular case of the generalized gamma distribution (Stacy,  1962). Then, after some algebra, Equation (\ref{eqmix2}) reduces to the PDF in (\ref{RLPDF}).

\subsection{Quantile function}
The quantile function (QF) of the GRL distribution defined in (\ref{RLPDF}), say $Q(p)$ where $0<p<1$, it can be obtained by solving the equation $F(Q(p))=p$ in (\ref{cdf}) for $Q(p)$ in terms of $p$, and this implies
\begin{align}\label{qf}
Q(p)=\left(-\lambda\left[ W_{-1} \left[(\lambda -1) (p-1) e^{1-\lambda }\right]+\lambda-1 \right ] \right)^{1/\alpha },
\end{align}
where $W_{-1}(x)$ is the negative branch of the Lambert function.

\subsection{Moments}

Moments play an important role in statistical theory, in this section we provide the $r-$th moment, the mean and variance for the GRL distribution.

\begin{proposition}\label{moments1} For the random variable $T$ follows the GRL distribution,
	the $r-$th moment is given by
	\begin{equation}\label{rmNM}
	\mu_r= E[T^r]=\frac{r\lambda^{\frac{r}{\alpha}}}{\alpha(\lambda-1)}\left(\lambda+\frac{r}{\alpha}-1 \right)\Gamma\left(\frac{r}{\alpha}\right)  \, , \quad \text{for} \quad r\in\mathbb{N} .
	\end{equation}
\end{proposition}
\begin{proof}Note that, the $r-$th moment for the random variable in (\ref{mixeq1}) is given by
	\[
	E[T^r;\alpha,\lambda]=\lambda^{\frac{r}{\alpha}}\Gamma\left(\frac{r}{\alpha}+j\right), \ \mbox{ for } \ j=1,2 .
	\]

	Since the GRL model can be expressed as a two-component mixture, as in (\ref{eqmix2}), so we have
	\[
	\begin{aligned}
	\mu_r&= E[T^r]=  \int_{0}^{\infty}t^rf(t|\alpha,\lambda)dt
	= pE[X^r;\alpha,\lambda]+(1-p)E[X^r;\alpha+1,\lambda] \\& = \left(\frac{\lambda-2}{\lambda-1}\right)\lambda^{\frac{r}{\alpha}}\Gamma\left(\frac{r}{\alpha}+1\right)+\frac{1}{(\lambda-1)}\lambda^{\frac{r}{\alpha}}\Gamma\left(\frac{r}{\alpha}+2\right)\\&= \frac{r\lambda^{\frac{r}{\alpha}}}{\alpha(\lambda-1)}\left(\lambda+\frac{r}{\alpha}-1 \right)\Gamma\left(\frac{r}{\alpha}\right)  \cdot
	\end{aligned}
	\]
\end{proof}

\begin{proposition} The random variable $T$ follows the PDF in (\ref{RLPDF}), its mean and variance, respectively, are  given by

\begin{equation}\label{lem2}
\mu=\frac{\lambda^{\frac{1}{\alpha}}}{\alpha(\lambda-1)}\Gamma\left(\frac{1}{\alpha}\right)\left(\lambda+\frac{1}{\alpha}-1 \right) \ \ \text{ and }	
\end{equation}

\begin{equation}
\sigma^2=\frac{\lambda^{\frac{2}{\alpha}}}{\alpha^2(\lambda-1)^2}\left[2\alpha(\lambda-1)\Gamma\left(\frac{2}{\alpha}\right)\left(\lambda+\frac{2}{\alpha}-1 \right)-\Gamma^2\left(\frac{1}{\alpha}\right)\left(\lambda+\frac{1}{\alpha}-1 \right)^2\right].
\end{equation}

\end{proposition}
\begin{proof} From (\ref{rmNM}) and considering $r=1$, it follows that $\mu_1=\mu$. The second result can be obtained by using $\sigma^2=E[T^2]-\mu^2$ and with some algebra the proof is completed.
\end{proof}

\begin{proposition} The $r-$th central moment for the GRL distribution is given by
	\begin{equation}\label{rcmNMp}
	\begin{aligned}
	M_r&= E[T-\mu]^r= \sum_{i=0}^{r}\binom{r}{i}(-\mu)^{r-i}E[T^i] \\ &= \sum_{i=0}^{r}\binom{r}{i}\left[-\frac{\lambda^{\frac{1}{\alpha}}\Gamma\left(\alpha^{-1}\right)}{\alpha(\lambda-1)}\left(\lambda+\frac{1}{\alpha}-1 \right) \right]^{r-i}\left[\frac{i\lambda^{\frac{i}{\alpha}}}{\alpha(\lambda-1)}\left(\lambda+\frac{i}{\alpha}-1 \right)\Gamma\left(\frac{i}{\alpha}\right) \right] .
	\end{aligned}
	\end{equation}
\end{proposition}
\begin{proof} The result follows directly from the Proposition \ref{moments1}.\end{proof}

The mean, variance, skewness and kurtosis of the GRL distribution are computed numerically for different values of the parameters $\lambda$ and $\alpha$, using the R software. Table~\ref{tab:tab1} displays these numerical values. From Table~\ref{tab:tab1} we can indicate that the skewness of the GRL distribution can range in the interval $(-0.68158, 5.17333)$. The spread for the GRL kurtosis is much larger ranging from 2.69447 to 52.6597. Further, the GRL model can be left skewed or right skewed. Hence, the GRL distribution is a flexible distribution which can be used in modelling skewed data.

\begin{table}[H]
\centering
\caption{Mean, variance, skewness and kurtosis of the GRL distribution for different values of the parameters $\lambda$ and $\alpha$}
 \label{tab:tab1}
\begin{tabular}{crrrr}
\\ \hline \hline \\
${\pmb \phi}^{\intercal}$                  & Mean      &Variance &  Skewness &  Kurtosis
 \\ \\ \hline \\
$(\lambda=2.0,\alpha=0.5)$ & 24.00 & 1344.00  & 4.30     & 37.41    \\
$(\lambda=2.0,\alpha=0.7)$ & 8.27  & 72.44    & 2.39     & 12.65    \\
$(\lambda=2.0,\alpha=2.5)$ & 1.64  & 0.23     & 0.20     & 2.89     \\
$(\lambda=2.0,\alpha=3.5)$ & 1.41  & 0.09     & -0.04    & 2.88     \\
$(\lambda=3.1,\alpha=0.5)$ & 37.52 & 5030.15  & 5.17     & 52.66    \\
$(\lambda=3.1,\alpha=0.7)$ & 10.71 & 186.23   & 2.85     & 16.24    \\
$(\lambda=3.1,\alpha=2.5)$ & 1.66  & 0.42     & 0.18     & 2.73     \\
$(\lambda=3.1,\alpha=3.1)$ & 1.49  & 0.23     & -0.03    & 2.71     \\
$(\lambda=4.0,\alpha=1.5)$ & 2.78  & 3.19     & 0.92     & 3.89     \\
$(\lambda=4.0,\alpha=3.5)$ & 1.46  & 0.20     & -0.07    & 2.71     \\
$(\lambda=4.0,\alpha=5.0)$ & 1.29  & 0.08     & -0.35    & 2.95     \\
$(\lambda=4.0,\alpha=10)$  & 1.13  & 0.02     & -0.73    & 3.74     \\
$(\lambda=5.5,\alpha=3.5)$ & 1.56  & 0.23     & -0.03    & 2.69     \\
$(\lambda=5.5,\alpha=5.0)$ & 1.35  & 0.09     & -0.30    & 2.90     \\
$(\lambda=5.5,\alpha=5.5)$ & 1.31  & 0.07     & -0.36    & 2.99     \\
$(\lambda=5.5,\alpha=10)$  & 1.15  & 0.02     & -0.68    & 3.64     \\ \\ \hline \hline
\end{tabular}
\end{table}

\subsection{Order statistics}\

Let $X_{1}, X_{2},\ldots X_{n}$ be a random sample from (\ref{RLPDF}) and $X_{1:n}\leq X_{2:n}\leq \ldots \leq X_{n:n}$  denote
the the corresponding order statistics. It is well known that the PDF and the CDF of  the of $r$-th  order statistic say $X_{r:n}$; $1\leq r \leq n$, respectively,  are given by
\begin{equation}\label{orderapdf}
\begin{aligned}
f_{r:n}(x)&=\frac{n!}{(r-1)!(n-r)!}[F(x)]^{r-1}[1-F(x)]^{n-r}f(x) \\
&=\frac{n!}{(r-1)!(n-r)!}\sum_{u=0}^{n-r}(-1)^{u}\left(\begin{array}{c}n-r\\u\end{array}\right)[F(x)]^{r-1+u}f(x)
\end{aligned}
\end{equation}
and
\begin{equation}\label{orderacdf}
\begin{aligned}
F_{r:n}(x)=\sum_{l=k}^{n}\left(\begin{array}{c}n\\l\end{array}\right)[F(x)]^{l}[1-F(x)]^{n-l}=\sum_{l=k}^{n}\sum_{u=0}^{n-r}(-1)^{u}\left(\begin{array}{c}n\\l\end{array}\right)\left(\begin{array}{c}n-r\\u\end{array}\right)[F(x)]^{l+u},
\end{aligned}
\end{equation}
for $k=1, 2, \ldots, n$. It follows, from (\ref{orderapdf}) and (\ref{orderacdf}), that the PDF and CDF of the $r$-th order statistic of the GRL reduce to

\begin{equation*}
\begin{aligned}
f_{r:n}(x)=&\frac{n!}{(\lambda-1)(r-1)!(n-r)!}\left(\lambda+\frac{t}{\lambda}-2\right)e^{-\frac{t}{\lambda}}\sum_{u=0}^{n-r}(-1)^{u}\left(\begin{array}{c}n-r\\u\end{array}\right)\times \\ & \left[1-\left(\frac{1}{\lambda-1}\right)
\left(\lambda-1+\frac{t^\alpha}{\lambda} \right)e^{- \frac{t^\alpha}{\lambda}} \right]^{r-1+u}
\end{aligned}
\end{equation*}
and
$$F_{r:n}(x)=\sum_{l=k}^{n}\sum_{u=0}^{n-r}(-1)^{u}\left(\begin{array}{c}n\\l\end{array}\right)\left(\begin{array}{c}n-r\\u\end{array}\right)\left[1-\left(\frac{1}{\lambda-1}\right)
\left(\lambda-1+\frac{t^\alpha}{\lambda} \right)e^{- \frac{t^\alpha}{\lambda}}\right]^{l+u}.$$

\section{Inference}
\label{sec3}
In this section, we estimate of the GRL parameters $\lambda$ and $\alpha$ using eight frequentist approaches. These methods are: the weighted least-squares (WLSE), ordinary least squares (OLSE), maximum likelihood (MLE),  maximum product of spacing (MPSE), Cram\'{e}r--von Mises (CVME), Anderson--Darling (ADE), Right-tail Anderson--Darling (RADE) and percentile based (PCE) estimators.

\subsection{Maximum likelihood estimation}
In this sub-section we present the maximum likelihood estimator (MLE) of the parameters $\lambda$ and $\alpha$ of the GRL distribution.

 Let $T_1,\ldots,T_n$ be a random sample such that $T$ has PDF given in (\ref{RLPDF}). In this case, for $\pmb {\phi}=(\lambda,\alpha)^{\intercal}$, the likelihood function from (\ref{RLPDF}) is given by
\begin{equation}\label{veroiNM}
L(\pmb {\phi};\pmb {t})=\frac{\alpha^{n}}{{\lambda^{n+1}(\lambda-1)}^n}\prod_{i=1}^n{t_i^{\alpha-1}}\prod_{i=1}^n\left(\lambda^2+t_i^\alpha-2\lambda\right)\exp\left(-\frac{1}{\lambda}\sum_{i=1}^n t_i^\alpha\right).
\end{equation}

The log-likelihood function $l(\pmb {\phi};\pmb{t})=\log{L(\pmb {\phi};\pmb{t})}$ is given by
\begin{equation}\label{verogNM2}
\begin{aligned}
l(\pmb {\phi};\pmb {t})=&\ n\log\left(\alpha\right)-(n+1)\log(\lambda)-n\log(\lambda-1)-\frac{1}{\lambda}\sum_{i=1}^n t_i^\alpha +(\alpha-1)\sum_{i=1}^{n}\log(t_i) \\ &+ \sum_{i=1}^n\log\left(\lambda^2+t_i^\alpha-2\lambda\right).
\end{aligned}
\end{equation}

From the expressions $\frac{\partial}{\partial \lambda}l(\pmb {\phi};\pmb {t})=0$, $\frac{\partial}{\partial \alpha}l(\pmb {\phi};\pmb {t})=0$, we get the likelihood equations
\begin{equation*}\label{verogg21}
-\frac{n+1}{\lambda}-\frac{n}{\lambda-1}+\frac{1}{\lambda^2}\sum_{i=1}^n t_i^\alpha + \sum_{i=1}^n\frac{2(\lambda-1)}{\lambda^2+t_i^\alpha-2\lambda}=0
\end{equation*}
and
\begin{equation*}\label{verogg22}
\frac{n}{\alpha}-\frac{1}{\lambda}\sum_{i=1}^n t_i^\alpha\log(t_i)+\sum_{i=1}^{n}\log(t_i)+ \sum_{i=1}^n\frac{t_i^\alpha\log(t_i)}{\lambda^2+t_i^\alpha-2\lambda}=0 .
\end{equation*}

Under mild conditions (Migon,  2014) the ML estimates are asymptotically normal distributed with a bivariate normal distribution given by
\begin{equation*}
(\hat\lambda,\hat\alpha) \sim N_2[(\lambda,\alpha),H^{-1}(\lambda,\alpha)] \mbox{ for } n \to \infty , \end{equation*}
where the elements of the observed Fisher information matrix $H(\lambda,\alpha)$ are given by
\begin{equation*}
h_{11}(\lambda,\alpha)=-\frac{n+1}{\lambda^2}-\frac{n}{(\lambda-1)^2}+\frac{2}{\lambda^3}\sum_{i=1}^n t_i^\alpha - \sum_{i=1}^n\frac{2(t_i^\alpha-\lambda^2+2\lambda-2)}{\left(\lambda^2+t_i^\alpha-2\lambda\right)^2},
\end{equation*}
\begin{equation*}
h_{12}(\lambda,\alpha)=h_{21}(\alpha,\lambda)=-\frac{1}{\lambda^2}\sum_{i=1}^n t_i^\alpha\log(t_i)+\sum_{i=1}^n\frac{2(\lambda-1)t_i^\alpha\log(t_i)}{\left(\lambda^2+t_i^\alpha-2\lambda\right)^2},
\end{equation*}
\begin{equation*}
h_{22}(\lambda,\alpha)=+\frac{n}{\alpha^2}+\frac{1}{\lambda}\sum_{i=1}^n t_i^\alpha\log(t_i)^2 -\sum_{i=1}^n\frac{\lambda(\lambda-2)t_i^\alpha\log(t_i)^2}{\left(\lambda^2+t_i^\alpha-2\lambda\right)^2}.
\end{equation*}

This can also be done by using different
programs namely \texttt{R} (\texttt{optim} function), \texttt{SAS} (\texttt{PROC NLMIXED}) or by solving the nonlinear likelihood equations obtained by
differentiating $\ell $.

\subsection{Ordinary and weighted least-square estimators}

Let $x_{(1)},x_{(2)},\cdots ,x_{(n)}$ be the order statistics of the random
sample of size $n$ from $F\left( \mathbf{x};\lambda,\alpha \right) $ in (\ref{cdf}). The
ordinary least square estimators (OLSEs) (Swain et al., 1988).
$\widehat{\lambda }_{LSE}$
 and $\widehat{\alpha }_{LSE}$ can be obtained by
minimizing
\begin{equation*}
V\left(\lambda,\alpha \right) =\sum_{i=1}^{n}\left[ F\left(
x_{(i)}|\lambda,\alpha \right) -\frac{i}{n+1}\right] ^{2},
\end{equation*}%
with respect to $\lambda$ and $\alpha$. Or equivalently, the OLSEs follow
by solving the non-linear equations%
\begin{equation*}
\sum_{i=1}^{n}\left[ F\left( x_{(i)}|\lambda,\alpha \right) -\frac{i}{n+1}%
\right] \Delta _{s}\left( x_{(i)}|\lambda,\alpha \right) =0,~\ s=1,2,
\end{equation*}%
where
\begin{eqnarray}
\Delta _{1}\left( x_{(i)}|\lambda,\alpha \right) &=&\frac{\partial }{%
\partial \lambda }F\left( x_{(i)}|\lambda,\alpha \right) \text{ and } ~\Delta
_{2}\left( x_{(i)}|\lambda,\alpha \right) =\frac{\partial }{\partial
\alpha }F\left( x_{(i)}|\lambda,\alpha \right).  \label{LSEs}
\end{eqnarray}%
Note that the solution of $\Delta _{s}$ for $s=1,2$ can be obtained
numerically.

The weighted least-squares estimators (WLSEs) (Swain et al., 1988), $%
\widehat{\lambda }_{WLSE}$ and $\widehat{\alpha }_{WLSE}$, can be obtained by minimizing the following equation
\begin{equation*}
W\left(\lambda,\alpha \right) =\sum_{i=1}^{n}\frac{\left( n+1\right)
^{2}\left( n+2\right) }{i\left( n-i+1\right) }\left[ F\left( x_{(i)}|\lambda,\alpha \right) -\frac{i}{n+1}\right] ^{2}.
\end{equation*}%
Further, the WLSEs can also be derived by solving the non-linear equations
\begin{equation*}
\sum_{i=1}^{n}\frac{\left( n+1\right) ^{2}\left( n+2\right) }{i\left(
n-i+1\right) }\left[ F\left( x_{(i)}|\lambda,\alpha \right) -\frac{i}{n+1}%
\right] \Delta _{s}\left( x_{(i)}|\lambda,\alpha \right) =0,~\ s=1,2,
\end{equation*}%
where $\Delta _{1}\left( \cdot |\lambda,\alpha \right) $ and $\Delta
_{2}\left( \cdot |\lambda,\alpha \right)$ are provided in (\ref{LSEs}).

\subsection{Method of maximum product of spacing}

The maximum product of spacings (MPS) method (Cheng and Amin, 1979 and Cheng and Amin, 1983 and Ranneby, 1984), as an approximation to the Kullback-Leibler information
measure, is a good alternative to the MLE method.

Let $D_{i}(\lambda,\alpha)=F\left( x_{(i)}|\lambda,\alpha \right)
-F\left( x_{(i-1)}|\lambda,\alpha \right) $, for $i=1,2,\ldots ,n+1,$ be
the uniform spacing of a random sample from the {GRL} distribution, where
$F\left( x_{(0)}|\lambda,\alpha \right) =0$, $F\left( x_{(n+1)}|\lambda,\alpha \right) =1$ and \ $\sum_{i=1}^{n+1}D_{i}(\lambda,\alpha)=1$.
The maximum product of spacing estimators (MPSEs) for $\widehat{\lambda }%
_{MPSE}$ and  $\widehat{\alpha }_{MPSE}$ can be
obtained by maximizing the geometric mean of the spacing
\begin{equation*}
G\left(\lambda,\alpha \right) =\left[ \prod\limits_{i=1}^{n+1}D_{i}(%
\lambda,\alpha)\right] ^{\frac{1}{n+1}}
\end{equation*}%
with respect to $\lambda$ and $\alpha$, or, equivalently, by maximizing
the logarithm of the geometric mean of sample spacings
\begin{equation*}
H\left(\lambda,\alpha \right) =\frac{1}{n+1}\sum_{i=1}^{n+1}\log
D_{i}(\lambda,\alpha).
\end{equation*}%
The MPSEs of the GRL parameters can be obtained by solving the nonlinear
equations defined by%
\begin{equation*}
\frac{1}{n+1}\sum\limits_{i=1}^{n+1}\frac{1}{D_{i}(\lambda,\alpha)}\left[
\Delta _{s}(x_{(i)}|\lambda,\alpha)-\Delta _{s}(x_{(i-1)}|\lambda,\alpha)\right] =0,\quad s=1,2,
\end{equation*}%
where $\Delta _{1}\left( \cdot |\lambda,\alpha \right) $ and $\Delta
_{2}\left( \cdot |\lambda,\alpha \right) $ are defined in (\ref{LSEs}).

\subsection{The Cram\'{e}r--von Mises minimum distance estimators}
The Cram\'{e}r--von Mises estimators (CVMEs) as a type of minimum distance estimators have less bias than the other minimum distance estimators (Macdonald, 1971).
The CVMEs are obtained based on the difference between the estimates of the CDF and the empirical distribution function (Luceño, 2006).
The CVMEs of the GRL parameters are obtained by minimizing
\begin{equation*}
C(\lambda,\alpha)=\frac{1}{12n}+\sum_{i=1}^{n}\left[ F\left(
x_{(i)}|\lambda,\alpha \right) -{\frac{2i-1}{2n}}\right] ^{2},
\end{equation*}%
with respect to $\lambda$ and $\alpha$. Further, the CVMEs follow by
solving the non-linear equations
\begin{equation*}
\sum_{i=1}^{n}\left[ F\left( x_{(i)}|\lambda,\alpha \right) -{\frac{2i-1}{%
2n}}\right] \Delta _{s}\left( x_{(i)}|\lambda,\alpha \right) =0,\quad
s=1,2,
\end{equation*}%
where $\Delta _{1}\left( \cdot |\lambda,\alpha \right) $ and $\Delta
_{2}\left( \cdot |\lambda,\alpha \right)$ are provided in (\ref{LSEs}).

\subsection{The Anderson-Darling and right-tail Anderson-Darling estimators}

The Anderson-Darling statistic or Anderson-Darling estimator is another type of minimum distance estimators. The Anderson-Darling estimators (ADEs) of the GRL parameters are obtained by minimizing
\begin{equation*}
A(\lambda,\alpha)=-n-\frac{1}{n}\sum_{i=1}^{n}\left( 2i-1\right) \,\left[
\log F\left( x_{(i)}|\lambda,\alpha \right) +\log S\left( x_{(i)}|\lambda,\alpha \right) \right] ,
\end{equation*}%
with respect to $\lambda$ and $\alpha$. These ADEs can also be obtained
by solving the non-linear equations
\begin{equation*}
\sum_{i=1}^{n}\left( 2i-1\right) \left[ \frac{\Delta _{s}\left(
x_{(i)}|\lambda,\alpha \right) }{F\left( x_{(i)}|\lambda,\alpha \right) }%
-\frac{\Delta _{j}\left( x_{(n+1-i)}|\lambda,\alpha \right) }{S\left(
x_{(n+1-i)}|\lambda,\alpha \right) }\right] =0,\ \ s=1,2.
\end{equation*}%
The right-tail Anderson-Darling estimators (RADEs) of the GRL parameters
are obtained by minimizing
\begin{equation*}
R(\lambda,\alpha)=\frac{n}{2}-2\sum_{i=1}^{n}F\left( x_{i:n}|\lambda,\alpha \right) -\frac{1}{n}\sum_{i=1}^{n}\left( 2i-1\right) \log S\left(
x_{n+1-i:n}|\lambda,\alpha \right) ,
\end{equation*}%
with respect to $\lambda$ and $\alpha$. The RADEs can also be obtained by
solving the non-linear equations
\begin{equation*}
-2\sum_{i=1}^{n}\Delta _{s}\left( x_{i:n}|\lambda,\alpha \right) +\frac{1}{%
n}\sum_{i=1}^{n}\left( 2i-1\right) \frac{\Delta _{s}\left(
x_{_{n+1-i:n}}|\lambda,\alpha \right) }{S\left( x_{n+1-i:n}|\lambda,\alpha \right) }=0,\ \ s=1,2.
\end{equation*}%
where $\Delta _{1}\left( \cdot |\lambda,\alpha \right) $ and $\Delta
_{2}\left( \cdot |\lambda,\alpha \right)$ are defined in Equation (\ref{LSEs}).

\subsection{Method of percentile estimation}

This method was originally suggested by  Kao (1958, 1959).  Let $ u_{i}=i/\left( n+1\right) $ be an unbiased estimator of $F\left(x_{(i)}|\lambda,\alpha \right) $. Then, the PCE of the parameters of GRL distribution are
obtained by minimizing the following function%
\begin{equation*}
P(\lambda,\alpha)=\sum_{i=1}^{n}\left( x_{(i)}-\left[-\lambda\left(W_{-1} \left((\lambda -1) (u_{i}-1) e^{1-\lambda }\right)+\lambda-1 \right) \right]^{1/\alpha } \right) ^{2},
\end{equation*}%
with respect to $\lambda$ and $\alpha$, where $W_{-1}(x)$ is the negative branch of the Lambert function.

\section{Simulation analysis}
\label{sec4}

A simulation study is conducted to evaluate and compare the behavior of the estimates with respect to their: average of absolute value of biases ($|Bias(\widehat{\pmb \phi})|$),
$|Bias(\widehat{\pmb \phi})|=$ $\frac{1}{N}\sum_{i=1}^{N}|\widehat{\pmb \phi}-\pmb \phi|$,  the average of mean square errors (MSEs),
 $MSEs=\frac{1}{N}\sum_{i=1}^{N}(\widehat{\pmb \phi}-\pmb \phi)^2$, and average of mean relative errors (MREs),
$MREs=\frac{1}{N}\sum_{i=1}^{N}|\widehat{\pmb \phi}-\pmb \phi|/\pmb \phi$.\\ We generate $N=5000$ random samples: ${X_1,X_2,\ldots,X_N}$ of sizes $n=30, 50, 80,100$ and $200$ from GRL model by using equation (\ref{qf}) by choosing $\lambda=\{2.0,4.5\}$ and $\alpha=\{0.5,2.5,0.7,3.5\}$, we used R software (version 3.6.1) (R Core Team, 2019)
For each parameters combination and each sample, we estimate of the GRL parameters $\lambda$ and $\alpha$ using eight frequentist approaches: WLSE (weighted least-squares), OLSE (ordinary least squares), MLE (maximum likelihood), MPSE (maximum product of spacing), CVME (Cramer-von Mises), ADE (Anderson-Darling), RADE (right-tail Anderson-Darling) and PCE (percentile based). Then, the MSEs and MREs of the parameter estimates are computed. Simulated outcomes are listed in Tables~\ref{tab:tab2}-\ref{tab:tab9} (see Appendix A). Furthermore, these tables show the rank of each of the estimators among all the estimators in each row, which is the superscript indicators, and the $\sum Ranks$, which is the partial sum of the ranks for each column in a certain sample size.
Table~\ref{tab:tab11} shows the partial and overall rank of the estimators.

From tables~\ref{tab:tab2}-\ref{tab:tab9}, we observe that:

\begin{itemize}
 \item All Estimation methods show the property of consistency i.e., the MSEs and MREs decrease as sample size increase, for all parameter combinations, except the estimator method WLSE.
\\
\item WLSE Estimation method shows the property of consistency for all parameter combinations, except the combinations $\phi=(\lambda=3.1,\alpha=0.7)^T$ and $\phi=(\lambda=3.1,\alpha=3.5)^T$, for the parameter $\lambda$.
\end{itemize}

Form Table~\ref{tab:tab10}, and for the parameter combinations, we can conclude that the MPSE estimator method outperforms all the other estimator methods (overall score of 62). Therefore, depends on our study,
we can confirm the superiority of MPSE and ADE estimator methods for the GRL distribution.

\setcounter{table}{9} \renewcommand{\thetable}{\arabic{table}}

\begin{table}[tbp]
\centering
\caption{Partial and overall ranks of all the methods of estimation for various combination of $\pmb \phi$}
 \label{tab:tab10}
\resizebox{\textwidth}{!}{
\begin{tabular}{cccccccccc}
\hline \hline \\
$\pmb \phi^{\intercal}$  & \textbf{$n$} & \textbf{WLSE} & \textbf{OLSE} & \textbf{MLE} & \textbf{MPSE} & \textbf{CVME} & \textbf{ADE} & \textbf{RADE} & \textbf{PCE} \\ \\ \hline \\
                           & 30                            & 5.5                            & 7.5                            & 2                             & 3                              & 5.5                            & 1                             & 4                              & 7.5                           \\
                           & 50                            & 4                              & 8                              & 2                             & 3                              & 5                              & 1                             & 6                              & 7                             \\
$(\lambda=2.0,\alpha=0.5)$ & 80                            & 4                              & 7                              & 1                             & 2                              & 5                              & 3                             & 6                              & 8                             \\
                           & 120                           & 3                              & 6.5                            & 2                             & 1                              & 6.5                            & 4                             & 5                              & 8                             \\
                           & 200                           & 4                              & 7                              & 1.5                           & 1.5                            & 5                              & 3                             & 6                              & 8                             \\
                           &                               &                                &                                &                               &                                &                                &                               &                                &                               \\
                           & 30                            & 5.5                            & 8                              & 3                             & 4                              & 7                              & 1                             & 5.5                            & 2                             \\
                           & 50                            & 5                              & 8                              & 3.5                           & 2                              & 6                              & 1                             & 7                              & 3.5                           \\
$(\lambda=2.0,\alpha=2.5)$ & 80                            & 5                              & 8                              & 2.5                           & 1                              & 6                              & 4                             & 7                              & 2.5                           \\
                           & 120                           & 3                              & 8                              & 2                             & 1                              & 6                              & 4.5                           & 7                              & 4.5                           \\
                           & 200                           & 4                              & 7.5                            & 2                             & 1                              & 7.5                            & 5                             & 6                              & 3                             \\
                           &                               &                                &                                &                               &                                &                                &                               &                                &                               \\
                           & 30                            & 3                              & 8                              & 5                             & 2                              & 6                              & 1                             & 4                              & 7                             \\
                           & 50                            & 1                              & 8                              & 4                             & 2                              & 5.5                            & 3                             & 5.5                            & 7                             \\
$(\lambda=2.0,\alpha=0.7)$ & 80                            & 2                              & 8                              & 3                             & 1                              & 6                              & 4                             & 5                              & 7                             \\
                           & 120                           & 3                              & 7.5                            & 1                             & 2                              & 7.5                            & 4                             & 5                              & 6                             \\
                           & 200                           & 3                              & 8                              & 2                             & 1                              & 6                              & 4                             & 5                              & 7                             \\
                           &                               &                                &                                &                               &                                &                                &                               &                                &                               \\
                           & 30                            & 2                              & 8                              & 4                             & 5                              & 6                              & 1                             & 7                              & 3                             \\
                           & 50                            & 1                              & 8                              & 5                             & 4                              & 6.5                            & 2.5                           & 6.5                            & 2.5                           \\
$(\lambda=2.0,\alpha=3.5)$ & 80                            & 1                              & 8                              & 3.5                           & 2                              & 6.5                            & 5                             & 6.5                            & 3.5                           \\
                           & 120                           & 2.5                            & 8                              & 1                             & 2.5                            & 6                              & 4                             & 7                              & 5                             \\
                           & 200                           & 3.5                            & 7.5                            & 2                             & 1                              & 7.5                            & 5                             & 6                              & 3.5                           \\
                           &                               &                                &                                &                               &                                &                                &                               &                                &                               \\
                           & 30                            & 6                              & 3.5                            & 7                             & 1                              & 5                              & 2                             & 3.5                            & 8                             \\
                           & 50                            & 5                              & 3                              & 7                             & 1                              & 6                              & 2                             & 4                              & 8                             \\
$(\lambda=3.1,\alpha=0.5)$ & 80                            & 5.5                            & 3.5                            & 7                             & 1                              & 5.5                            & 2                             & 3.5                            & 8                             \\
                           & 120                           & 5                              & 3                              & 7                             & 1                              & 6                              & 2                             & 4                              & 8                             \\
                           & 200                           & 5                              & 3                              & 7                             & 1                              & 5                              & 2                             & 5                              & 8                             \\
                           &                               &                                &                                &                               &                                &                                &                               &                                &                               
 \\ \hline
\end{tabular}}
\end{table}

\begin{table}[tbp]
\centering
\caption*{Table 10: Partial and overall ranks of all the methods of estimation for various combination of $\pmb \phi$ (Continued)}
\resizebox{\textwidth}{!}{
\begin{tabular}{cccccccccc}
\hline \hline \\
$\pmb \phi^{\intercal}$  & \textbf{$n$} & \textbf{WLSE} & \textbf{OLSE} & \textbf{MLE} & \textbf{MPSE} & \textbf{CVME} & \textbf{ADE} & \textbf{RADE} & \textbf{PCE} \\ \\ \hline \\                           
                           & 30                            & 2                              & 5                              & 8                             & 1                              & 6                              & 3                             & 7                              & 4                             \\
                           & 50                            & 3                              & 4                              & 8                             & 1                              & 6                              & 5                             & 7                              & 2                             \\
$(\lambda=3.1,\alpha=2.5)$ & 80                            & 4                              & 3                              & 8                             & 1                              & 6                              & 5                             & 7                              & 2                             \\
                           & 120                           & 6                              & 3                              & 8                             & 1                              & 4.5                            & 2                             & 7                              & 4.5                           \\
                           & 200                           & 5.5                            & 2.5                            & 8                             & 1                              & 5.5                            & 2.5                           & 7                              & 4                             \\
                           &                               &                                &                                &                               &                                &                                &                               &                                &                               \\
                           & 30                            & 3.5                            & 2                              & 8                             & 1                              & 5                              & 6                             & 3.5                            & 7                             \\
                           & 50                            & 4                              & 3                              & 8                             & 1                              & 5                              & 6                             & 2                              & 7                             \\
$(\lambda=3.1,\alpha=0.7)$ & 80                            & 4                              & 3                              & 7                             & 1                              & 5                              & 6                             & 2                              & 8                             \\
                           & 120                           & 4                              & 3                              & 7                             & 1                              & 5                              & 6                             & 2                              & 8                             \\
                           & 200                           & 5                              & 2                              & 7                             & 1                              & 6                              & 4                             & 3                              & 8                             \\
                           &                               &                                &                                &                               &                                &                                &                               &                                &                               \\
                           & 30                            & 2                              & 6                              & 7                             & 1                              & 8                              & 4                             & 5                              & 3                             \\
                           & 50                            & 2                              & 5                              & 8                             & 1                              & 7                              & 4                             & 6                              & 3                             \\
$(\lambda=3.1,\alpha=3.5)$ & 80                            & 2                              & 4.5                            & 8                             & 1                              & 6                              & 4.5                           & 7                              & 3                             \\
                           & 120                           & 2                              & 3                              & 8                             & 1                              & 6                              & 4                             & 7                              & 5                             \\
                           & 200                           & 2                              & 3                              & 8                             & 1                              & 5                              & 4                             & 7                              & 6                             \\
                           &                               &                                &                                &                               &                                &                                &                               &                                &                               \\
$\sum$ Ranks               &                               & 142.5                          & 222.5                          & 203                           & 62                             & 236.5                          & 137                           & 216.5                          & 220                           \\ \\
Overall Rank               &                               & 3                              & 7                              & 4                             & 1                              & 8                              & 2                             & 5                              & 6
\\ \\ \hline \hline
\end{tabular}}
\end{table}

\section{Data analysis}
\label{sec5}

In this section, we illustrate the importance of the GRL distribution in modelling skewed data using two real data sets from the medicine and geology fields. The first data represent the survival times, in weeks, of 33
patients suffering from acute myelogeneous leukaemia  (Feigl and Zelen, 1965). These data have been analyzed by Abouelmagd et al. (2018), Nassar et al. (2018) and Sen et al. (2019).
The second data is used to evaluate the risks associated with earthquakes occurring close to the central site of a nuclear power plant. This data set refers to the distances, in miles, to the nuclear power plant of the most recent 8 earthquakes of intensity larger than a given value (Castillo,  2012) and it consists of 60 observations.

We consider some measures of goodness-of-fit namely, minus maximized log-likelihood ($-\widehat{\ell }$), Cram\'{e}r-Von Mises ($W^{\ast }$), Anderson-Darling ($A^{\ast }$) and Kolmogorov Smirnov (KS) statistics with
its bootstrapped $p$-value (PV), to compare the fits of the GRL distribution with other
competitive models given in Tables~\ref{tab:tab12} and~\ref{tab:tab13}.  We draw $999999$ bootstrap samples to obtain the KS bootstrapped PV.

\bigskip

The fitted competitive models are namely given in~\ref{tab:tab11}, and their densities (for $x>0$) are given by:

MOEx: $f(x)=\alpha \lambda \mathrm{\exp }(-\lambda x)\left[ 1-\left(
1-\alpha \right) \mathrm{\exp }(-\lambda x)\right] ^{-2}.$

BEx: $f(x)=\frac{\lambda }{B(a,b)}\mathrm{\exp }(-b\lambda x)\left[ 1-%
\mathrm{\exp }(-\lambda x)\right] ^{a-1}.$

EEx: $f(x)=\alpha \lambda \mathrm{\exp }(-\lambda x)\left[ 1-\mathrm{\exp }%
(-\lambda x)\right] ^{\alpha -1}.$

Ga: $f(x)=\frac{b^{-a}}{\Gamma \left( a\right) }x^{a-1}\exp \left(
-x/b\right) .$

GLi: $f(x)=\frac{\alpha \lambda ^{2}}{\lambda +1}\left( 1+x\right) \mathrm{%
\exp }(-\lambda x)\left[ 1-\frac{1+\lambda +\lambda x}{\lambda +1}\mathrm{%
\exp }(-\lambda x)\right] ^{\alpha -1}.$

TTLi: $f\left( x\right) =\left(\frac{a^{2}}{\alpha +a}(1+\alpha x)\exp (-ax)\right)
\left(1+\lambda-2 \lambda \left(
1-\frac{\alpha +a+\alpha ax}{\alpha +a}\exp (-ax)\right) \right).$

PLx: $f\left( x\right) =\alpha \beta \lambda \left( 1+\beta x\right)
^{-\alpha -1}\mathrm{\exp }\left[ -\lambda \left( 1+\beta x\right) ^{-\alpha
}\right] \left[ 1-\mathrm{\exp }(-\lambda )\right] ^{-1}.$

LiGc: $f\left( x\right) =\left[ 1-\left( 1+\frac{ax}{a+1}\right) \exp (-ax)%
\right] /\left[ 1-\alpha \left( 1+\frac{ax}{a+1}\right) \exp (-ax)\right] .$

Li: $f\left( x\right) =\frac{\lambda ^{2}}{\lambda +1}\left( 1+x\right)
\mathrm{\exp }(-\lambda x).$

The parameters of the above densities are all positive real numbers except $%
\left\vert \lambda \right\vert \leq 1$ for the TTLi distribution and $\alpha
\in \left( 0,1\right) $ for the LiGc distribution.

\begin{table}[tbp]
\centering
\caption{The fitted competitive models}
 \label{tab:tab11}
\resizebox{\textwidth}{!}{
\begin{tabular}{ll}
\hline \hline \\
Distribution & Author(s) \\ \\ \hline \\
Ramos-Louzada (RL) (Special case) & Ramos and Louzada (2019) \\
Marshall-Olkin exponential (MOEx) & Marshall and Olkin (1997) \\
Beta exponential (BEx) & Jones (2004) \\
Exponentiated exponential (EEx) & Gupta and Kundu (2001) \\
Gamma (Ga) & -- \\
Generalized Lindley (GLi) & Nadarajah et al. (2011) \\
Transmuted two-parameter Lindley (TTLi) & Kemaloglu and Yilmaz (2017) \\
Poisson-Lomax (PLx) & Al-Zahrani and Sagor (2014) \\
Lindley geometric (LiGc) & Zakerzadeh and Mahmoudi (2012) \\
Lindley (Li) & Lindley (1958) \\ \\ \hline\hline
\end{tabular}}
\end{table}

\begin{table}[tbp]
\centering
\caption{Goodness-of-fit statistics, MLEs and SEs for leukaemia data}
\label{tab:tab12}
\resizebox{\textwidth}{!}{
\begin{tabular}{lccccccl}
\hline \hline \\
Model & $-\widehat{\ell }$ & $W^{\ast}$ & $A^{\ast}$ & K-S     & Bootstrapped PV &                      & Estimates (SE)     \\ \\ \hline \\
GRL   & 153.58031          & 0.09469    & 0.65053    & 0.13637 & 0.15573         & $\widehat{\lambda }$ & 14.6996 (7.67698)   \\
      &                    &            &            &         &                 & $\widehat{\alpha }$  & 0.77410 (0.10927)   \\
MOEx  & 153.59511          & 0.09725    & 0.65062    & 0.14919 & 0.13112         & $\widehat{\alpha }$  & 0.30374 (0.21450)   \\
      &                    &            &            &         &                 & $\widehat{\lambda }$ & 0.01344 (0.00677)   \\
BEx   & 153.65124          & 0.09676    & 0.67015    & 0.13830 & 0.23219         & $\widehat{a}$        & 0.67358 (0.15375)   \\
      &                    &            &            &         &                 & $\widehat{b}$        & 0.79022 (1.59942)   \\
      &                    &            &            &         &                 & $\widehat{\lambda }$ & 0.02421 (0.05263)   \\
EEx   & 153.65164          & 0.09664    & 0.66905    & 0.13834 & 0.19718         & $\widehat{\alpha }$  & 0.67805 (0.14476)   \\
      &                    &            &            &         &                 & $\widehat{\lambda }$ & 0.01880 (0.00476)   \\
Ga    & 153.67366          & 0.09662    & 0.66842    & 0.13901 & 0.19166         & $\widehat{a}$        & 0.68776 (0.14403)   \\
      &                    &            &            &         &                 & $\widehat{b}$        & 59.4374 (17.6235) \\
GLi   & 154.70797          & 0.11038    & 0.76662    & 0.14254 & 0.30774         & $\widehat{\alpha }$  & 0.36486 (0.07752)   \\
      &                    &            &            &         &                 & $\widehat{\lambda }$ & 0.02603 (0.00606)   \\
TTLi  & 154.85217          & 0.09695    & 0.66683    & 0.20278 & 0.01875         & $\widehat{\alpha }$  & 0.00007 (0.00891)   \\
      &                    &            &            &         &                 & $\widehat{a}$        & 0.02071 (0.01043)   \\
      &                    &            &            &         &                 & $\widehat{\lambda }$ & 0.36761 (0.36088)   \\
PLx   & 155.92974          & 0.15942    & 0.95843    & 0.14182 & 0.10734         & $\widehat{\alpha }$  & 0.74788 (0.19910)   \\
      &                    &            &            &         &                 & $\widehat{\beta }$   & 2.38873 (8.77439)   \\
      &                    &            &            &         &                 & $\widehat{\lambda }$ & 9.70373 (19.3366)  \\
RL    & 155.45330          & 0.09726    & 0.67297    & 0.21831 & 0.00762         & $\widehat{\lambda }$ & 39.8689 (7.11473)  \\
LiGc  & 161.98422          & 0.13034    & 0.85743    & 0.24283 & 0.00001         & $\widehat{\alpha }$  & 0.91431 (0.07170)   \\
      &                    &            &            &         &                 & $\widehat{a}$        & 0.02303 (0.00888)   \\
Li    & 168.83368          & 0.11041    & 0.76655    & 0.32512 & 0.00001         & $\widehat{\lambda }$ & 0.04781 (0.00589)  \\ \\ \hline \hline
\end{tabular}}
\end{table}

\begin{table}[tbp]
\centering
\caption{Goodness-of-fit statistics, MLEs and SEs for epicenter data}
\label{tab:tab13}
\resizebox{\textwidth}{!}{
\begin{tabular}{lccccccl}
\hline \hline \\
Model & $-\widehat{\ell }$ & $W^{\ast}$ & $A^{\ast}$ & K-S     & Bootstrapped PV &                      & Estimates (SE)     \\ \\ \hline \\
GRL   & 323.74634          & 0.17905    & 1.04395    & 0.13180 & 0.12386         & $\widehat{\lambda }$ & 3728110.8 (8389) \\
      &                    &            &            &         &                 & $\widehat{\alpha }$  & 2.97246 (0.02463)          \\
Ga    & 325.52256          & 0.24286    & 1.40568    & 0.15528 & 0.00124         & $\widehat{a}$        & 6.18451 (1.10003)          \\
      &                    &            &            &         &                 & $\widehat{b}$        & 23.3718 (4.33053)         \\
BEx   & 326.12308          & 0.25892    & 1.50278    & 0.16002 & 0.00445         & $\widehat{a}$        & 7.13835 (1.56344)          \\
      &                    &            &            &         &                 & $\widehat{b}$        & 2.19531 (1.21066)          \\
      &                    &            &            &         &                 & $\widehat{\lambda }$ & 0.01133 (0.00433)          \\
GLi   & 326.54668          & 0.26913    & 1.56480    & 0.15616 & 0.00117         & $\widehat{\alpha }$  & 3.82547 (0.93381)          \\
      &                    &            &            &         &                 & $\widehat{\lambda }$ & 0.02387 (0.00247)          \\
EEx   & 326.98566          & 0.28122    & 1.63873    & 0.16018 & 0.00110         & $\widehat{\alpha }$  & 8.64049 (2.25238)          \\
      &                    &            &            &         &                 & $\widehat{\lambda }$ & 0.01913 (0.00216)          \\
MOEx  & 327.99000          & 0.22196    & 1.29842    & 0.13454 & 0.00156         & $\widehat{\alpha }$  & 30.3754 (12.0111)        \\
      &                    &            &            &         &                 & $\widehat{\lambda }$ & 0.02479 (0.00253)          \\
TTLi  & 329.75618          & 0.25369    & 1.47104    & 0.14685 & 0.01108         & $\widehat{\alpha }$  & 177.960 (3753.42)     \\
      &                    &            &            &         &                 & $\widehat{a}$        & 0.01876 (0.00518)          \\
      &                    &            &            &         &                 & $\widehat{\lambda }$ & -0.99999 (1.37433)         \\
PLx   & 331.17805          & 0.38592    & 2.27644    & 0.18476 & 0.00001         & $\widehat{\alpha }$  & 3.63997 (0.85273)          \\
      &                    &            &            &         &                 & $\widehat{\beta }$   & 0.01796 (0.01428)          \\
      &                    &            &            &         &                 & $\widehat{\lambda }$ & 53.4162 (59.5182)        \\
Li    & 340.54657          & 0.24249    & 1.40349    & 0.19602 & 0.00056         & $\widehat{\lambda }$ & 0.01374 (0.00125)      \\
LiGc  & 340.54657          & 0.24249    & 1.40349    & 0.19603 & 0.00009         & $\widehat{\alpha }$  & 0.00001 (0.41139)          \\
      &                    &            &            &         &                 & $\widehat{a}$        & 0.01374 (0.00289)          \\
RL    & 358.41349          & 0.24322    & 1.40786    & 0.33146 & 0.00000         & $\widehat{\lambda }$ & 143.524 (18.6612)       \\

  \\

 \hline \hline
\end{tabular}}
\end{table}

\begin{table}[tbp]
\centering
\caption{The parameter estimates under various estimation methods and the goodness-of-fit statistics
for leukaemia data}
 \label{tab:tab14}
\resizebox{\textwidth}{!}{
\begin{tabular}{lccccccc}
\hline \hline \\ \\
Method & $\hat \lambda$ & $\hat \alpha$ & $- \hat \ell$ & $W$     & $A$     & $K-S$   & Bootstrapped PV \\ \\ \hline \\
WLSE   & 10.92982       & 0.69340       & 153.92720     & 0.09412 & 0.64348 & 0.11093 & 0.36749         \\
OLSE   & 8.26873        & 0.62355       & 154.77563     & 0.09401 & 0.63966 & 0.09917 & 0.36766         \\
MLE    & 14.03083       & 0.76522       & 153.58430     & 0.09465 & 0.65005 & 0.13133 & 0.15634         \\
MPSE    & 11.97607 & 0.71768 & 153.75038 & 0.09424       &   0.64529 & 0.11903 & 0.55489         \\
CVME   & 9.09894        & 0.64955       & 154.37521     & 0.09403 & 0.64104 & 0.09309 & 0.32234         \\
ADE    & 10.34346       & 0.68310       & 153.99337     & 0.09411 & 0.64310 & 0.10464 & 0.38431         \\
RADE   & 10.39537       & 0.68317       & 154.00034     & 0.09411 & 0.64301 & 0.10549 & 0.43200         \\
PCE    & 24.31768       & 0.86231       & 154.07402     & 0.09492 & 0.65389 & 0.18744 & 0.08984
 \\  \\ \hline \hline
\end{tabular}}
\end{table}

\begin{table}[tbp]
\centering
\caption{The parameter estimates under various estimation methods and the goodness-of-fit statistics
for epicenter data}
 \label{tab:tab15}
\resizebox{\textwidth}{!}{
\begin{tabular}{lccccccc}
\hline \hline \\ \\
Method & $\hat \lambda$ & $\hat \alpha$ & $- \hat \ell$ & $W$     & $A$     & $K-S$   & Bootstrapped PV \\ \\ \hline \\
WLSE   & 1125196.57183  & 2.73252       & 324.154       & 0.18569 & 1.07845 & 0.11796 & 0.04276         \\
OLSE   & 267870.74460    & 2.44310       & 325.819       & 0.19434 & 1.12489 & 0.09761 & 0.32650         \\
MLE    & 3728110.76443  & 2.97253       & 323.746       & 0.17904 & 1.04392 & 0.13192 & 0.12335         \\
MPSE   & 1947372.47807  & 2.84345       & 323.858       & 0.18251 & 1.06180 & 0.12630 & 0.15794         \\
CVME   & 357125.58188   & 2.50034       & 325.407       & 0.19263 & 1.11557 & 0.10028 & 0.12011         \\
ADE    & 650623.70779   & 2.62361       & 324.591       & 0.18879 & 1.09491 & 0.11313 & 0.13961         \\
RADE   & 2842081.10850  & 2.90677       & 323.923       & 0.18149 & 1.05625 & 0.10840 & 0.19406         \\
PCE    & 2842081.10830  & 2.91651       & 323.778       & 0.18065 & 1.05215 & 0.12569 & 0.08414
 \\  \\ \hline \hline
\end{tabular}}
\end{table}

The numerical values of $-\widehat{\ell }$, $W^{\ast }$, $A^{\ast }$, KS, Bootstrapped PV, the MLEs and their corresponding standard errors (SEs) (given in parentheses) of the fitted models are listed in Tables~\ref{tab:tab12} and \ref{tab:tab13}, for both data sets, respectively. The figures in Tables~\ref{tab:tab12} and \ref{tab:tab13} show that the GRL distribution has the lowest values for all goodness-of-fit statistics among all fitted models. \\
Tables~\ref{tab:tab14} and ~\ref{tab:tab15} display the parameter estimates under various estimation methods and the goodness-of-fit statistics for both data sets, respectively. From Table~\ref{tab:tab14} and based on the $K-S$
bootstrapped PV, we recommend to use the MPSE method to estimate the parameters of the GRL distribution for leukaemia data, while the OLS method is recommended to estimate the GRL parameters for epicenter data, based on the $K-S$ bootstrapped PV in Table~\ref{tab:tab15}.

The histogram of the fitted GRL distribution and the other distributions are displayed in Figures~\ref{fig:fig3} and~\ref{fig:fig4} for the two data sets, respectively. Figures~\ref{fig:fig3} and~\ref{fig:fig4} show the plots of PDFs and CDFs of the fitted models for leukaemia and epicenter data. The HRF plot of the GRL distribution and the TTT plot of leukaemia data are displayed in Figure~\ref{fig:fig5}, whereas the HRF plot of the GRL distribution and the TTT plot of epicenter data are displayed in Figure~\ref{fig:fig6}. It is shown that, the HRF is decreasing for leukaemia data, whereas the HRF is increasing for epicenter data. Furthermore, the scaled TTT plot for the leukemia data is convex which indicates a decreasing HRF and it is concave for epicenter data which indicates an increasing HRF. Then, the GRL distribution is a suitable for modeling leukaemia and epicenter data.

\begin{figure}[H]
\begin{subfigure}{.5\textwidth}
  \centering
  \includegraphics[width=1\linewidth]{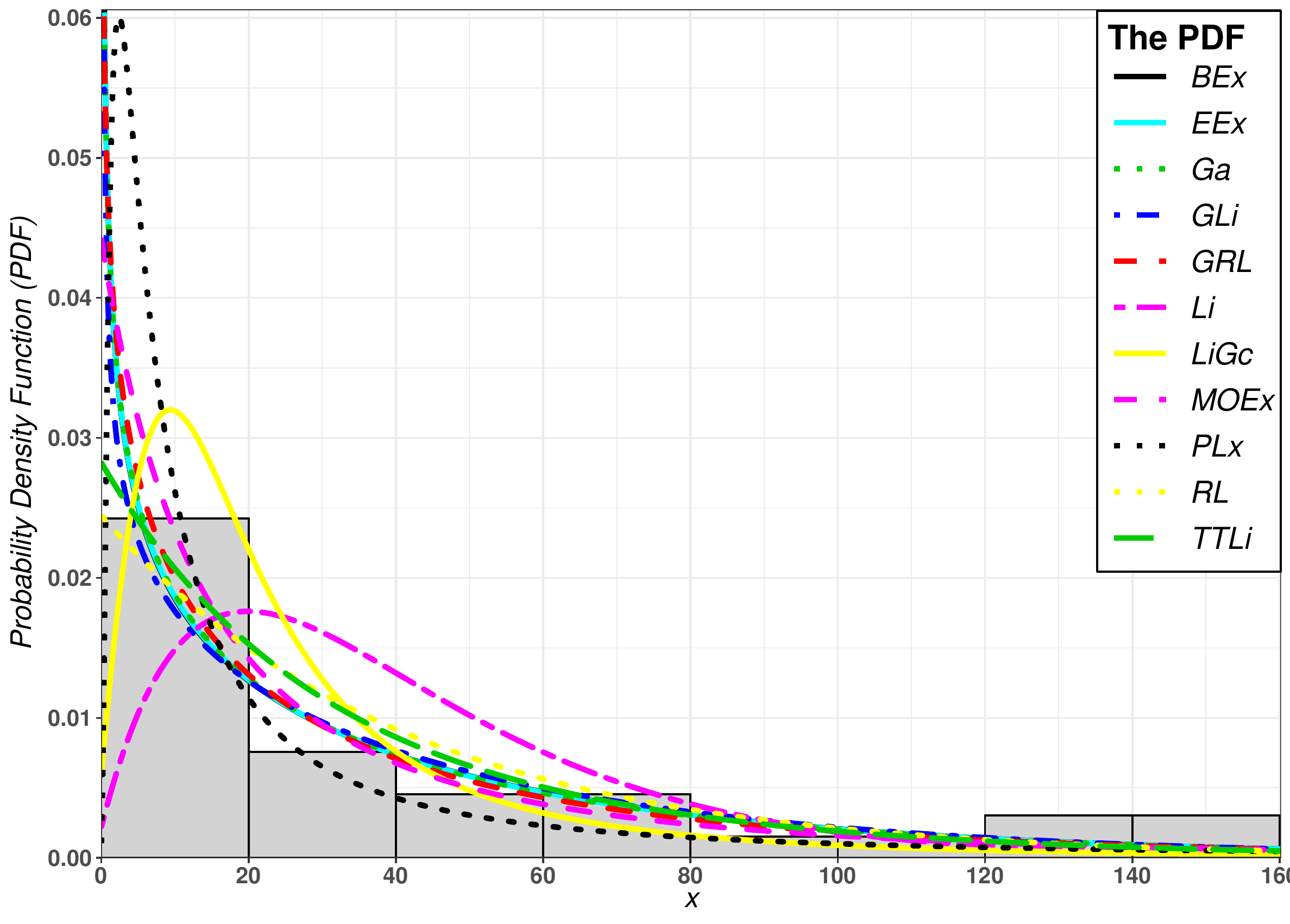}
  \caption{PDF plots}
\end{subfigure}
    \begin{subfigure}{.5\textwidth}
  \centering
  \includegraphics[width=1\linewidth]{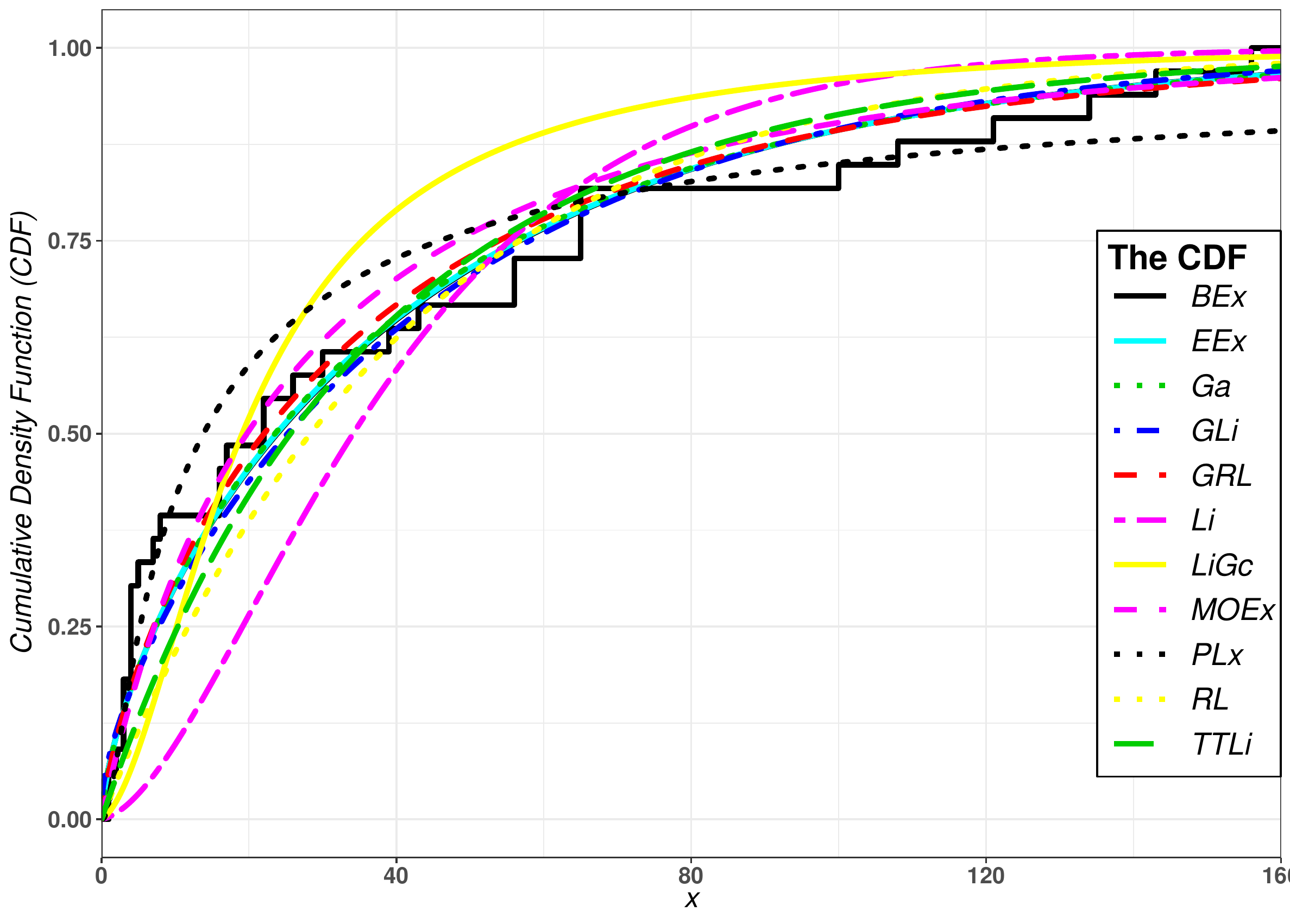}
  \caption{CDF plots}
\end{subfigure}
  \caption{PDFs and CDFs of the fitted models for leukaemia data}
  \label{fig:fig3}
\end{figure}

\begin{figure}[H]
\begin{subfigure}{.5\textwidth}
  \centering
  \includegraphics[width=1\linewidth]{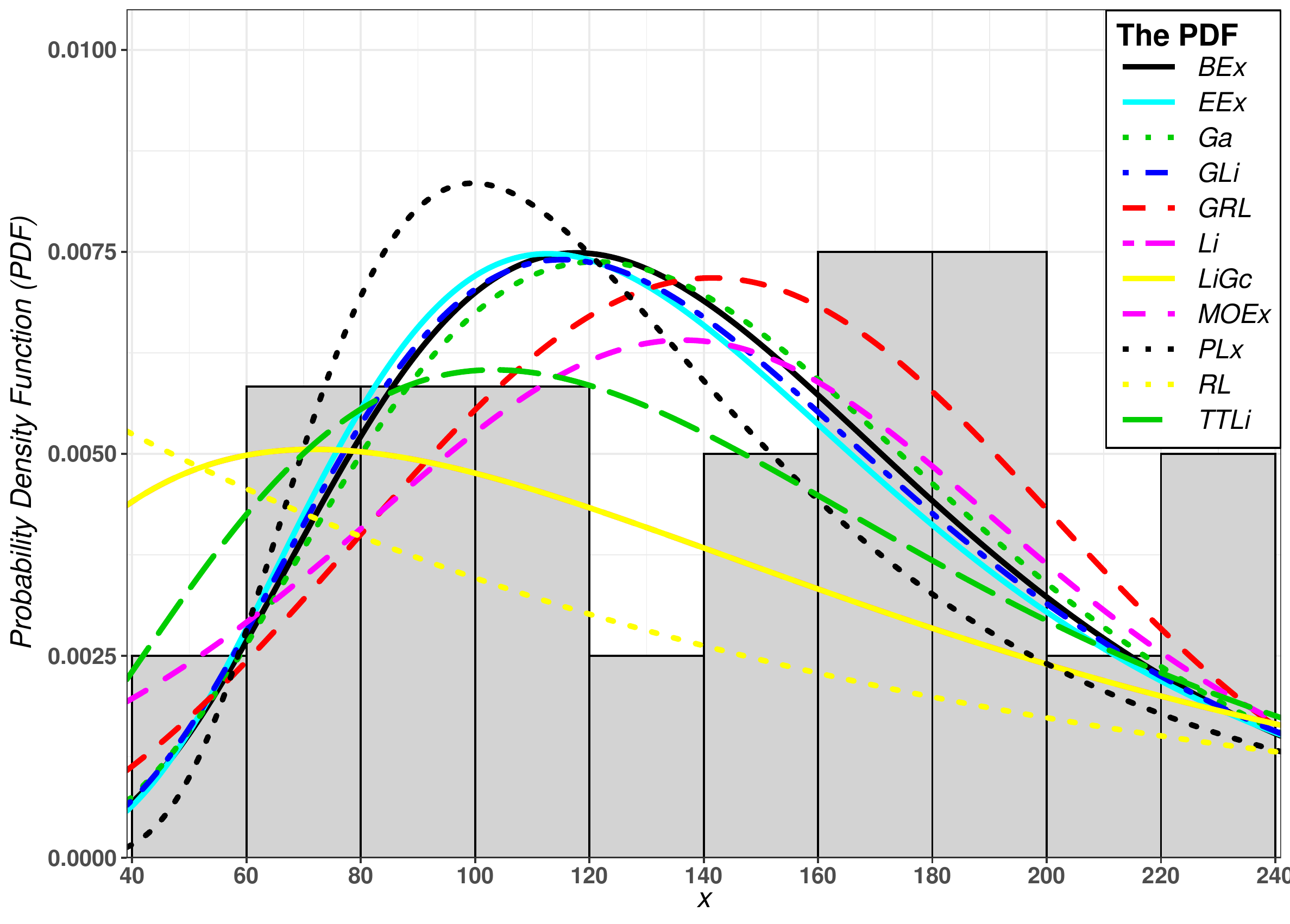}
  \caption{PDF plots}
\end{subfigure}
    \begin{subfigure}{.5\textwidth}
  \centering
  \includegraphics[width=1\linewidth]{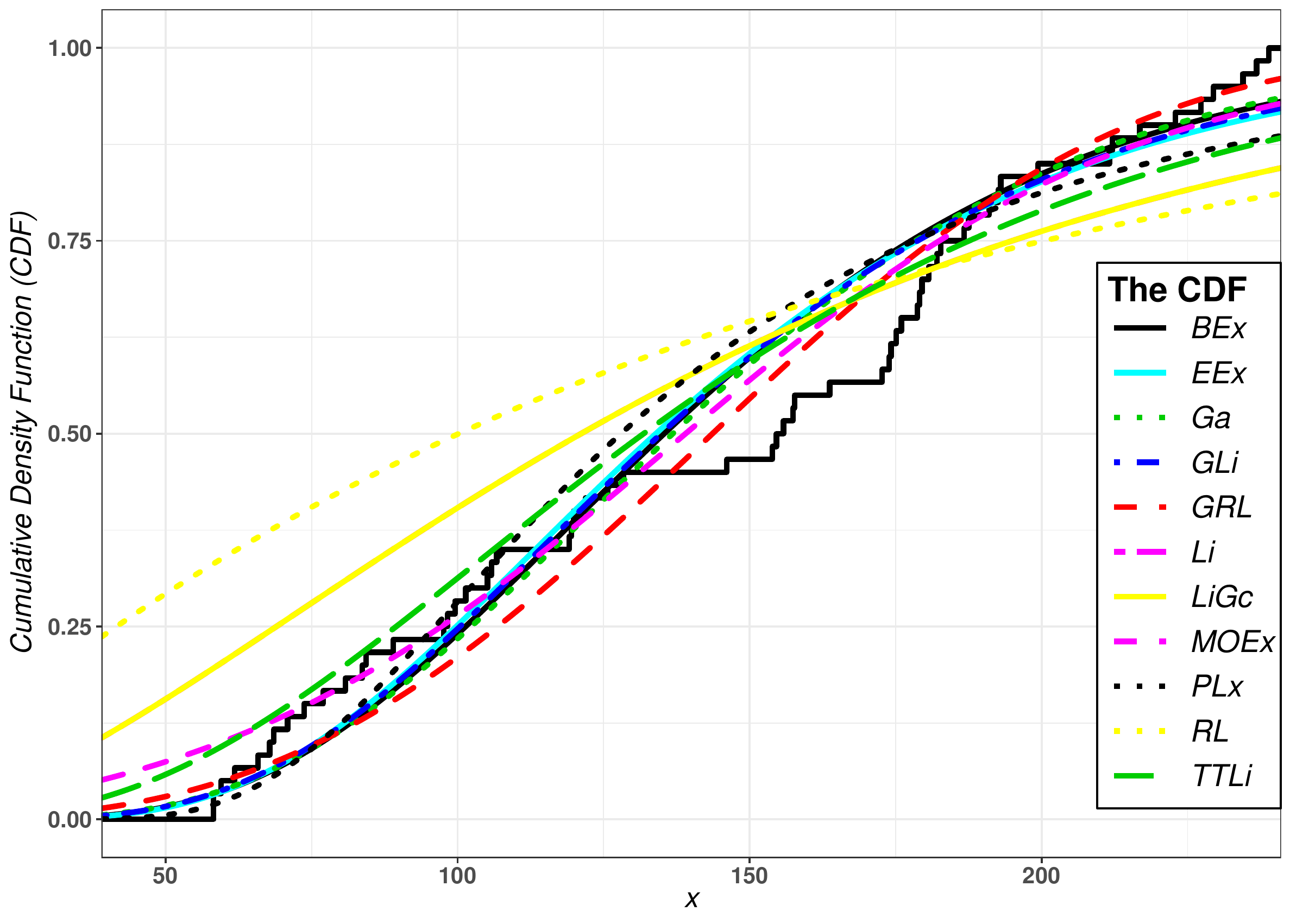}
  \caption{CDF plots}
\end{subfigure}
  \caption{PDFs and CDFs of the fitted models for epicenter data}
  \label{fig:fig4}
\end{figure}

\begin{figure}[H]
\begin{subfigure}{.5\textwidth}
  \centering
  \includegraphics[width=1\linewidth]{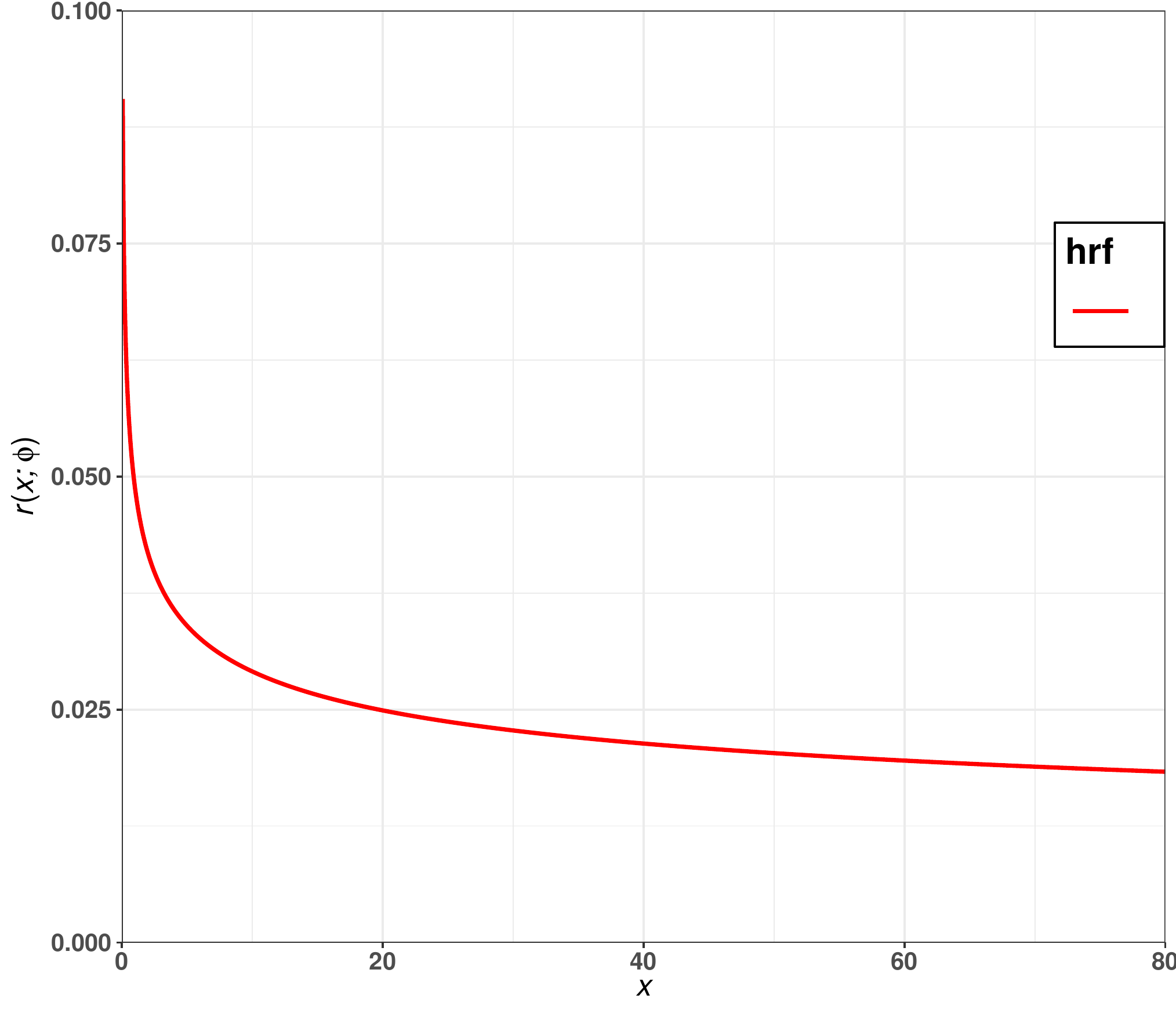}
  \caption{HRF plot}
\end{subfigure}
    \begin{subfigure}{.5\textwidth}
  \centering
  \includegraphics[width=1\linewidth]{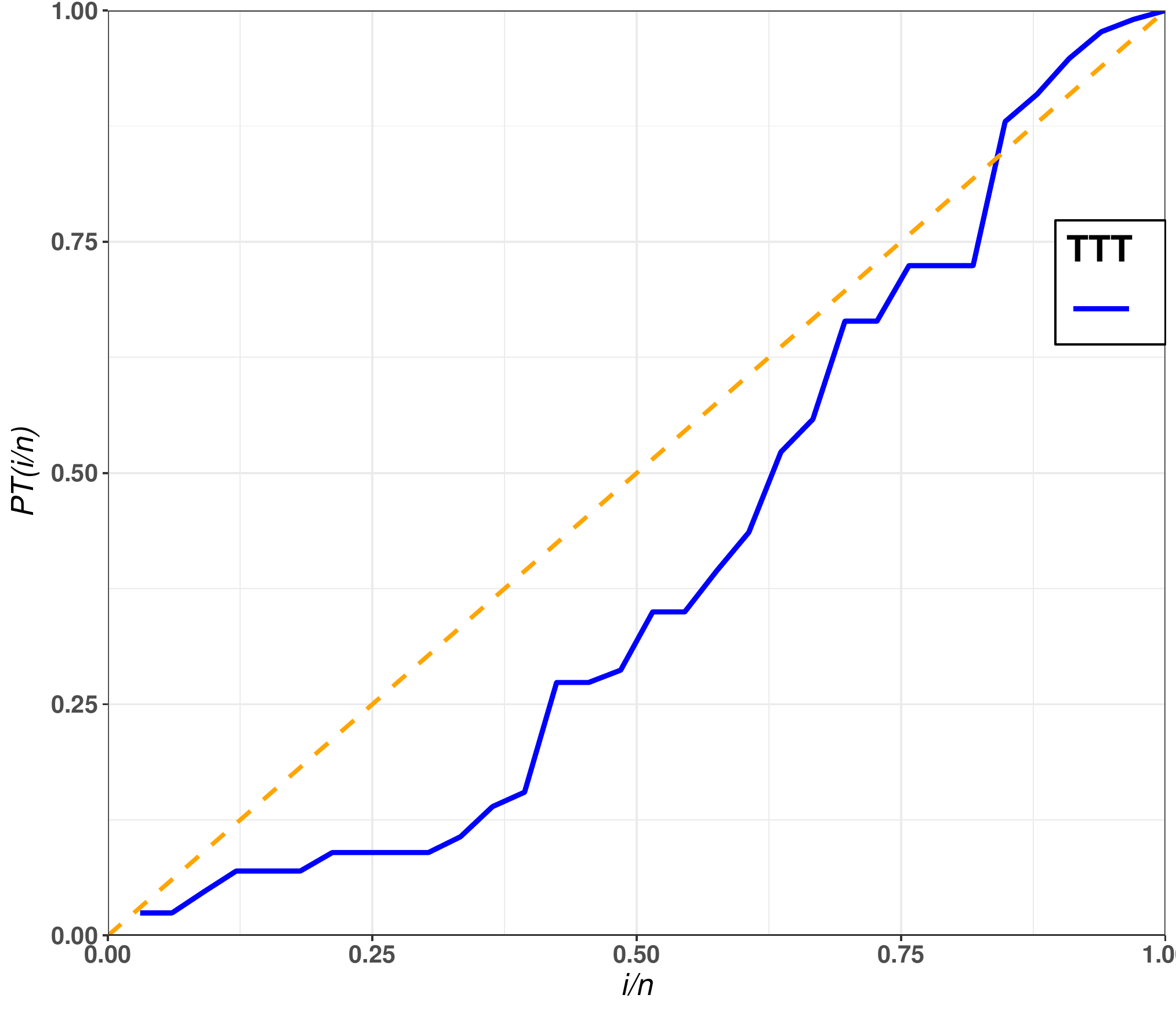}
  \caption{TTT plot}
\end{subfigure}
  \caption{The HRF plot of the GRL distribution and TTT plot for leukaemia data}
  \label{fig:fig5}
\end{figure}

\begin{figure}[H]
\begin{subfigure}{.5\textwidth}
  \centering
  \includegraphics[width=1\linewidth]{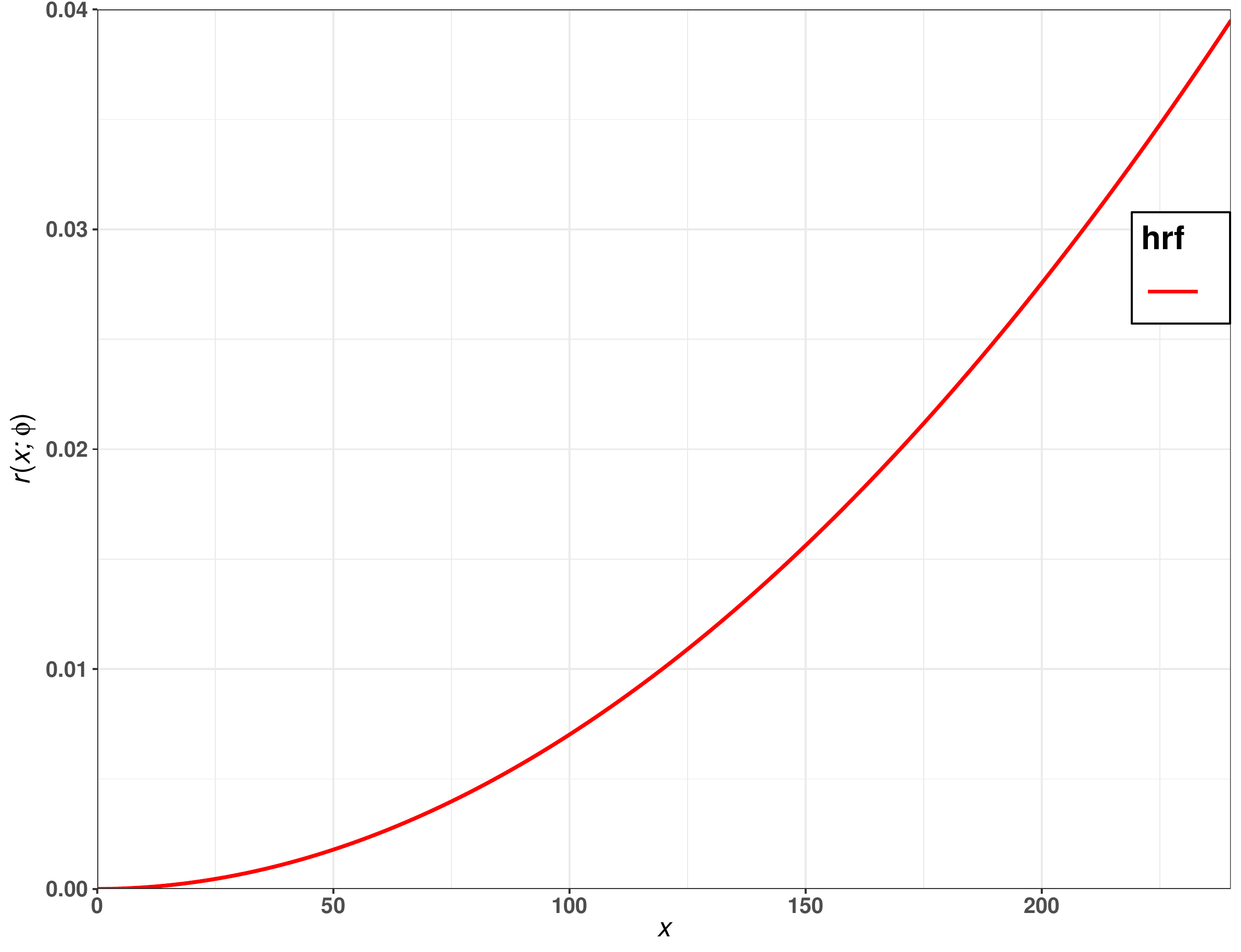}
  \caption{HRF plot}
\end{subfigure}
    \begin{subfigure}{.5\textwidth}
  \centering
  \includegraphics[width=1\linewidth]{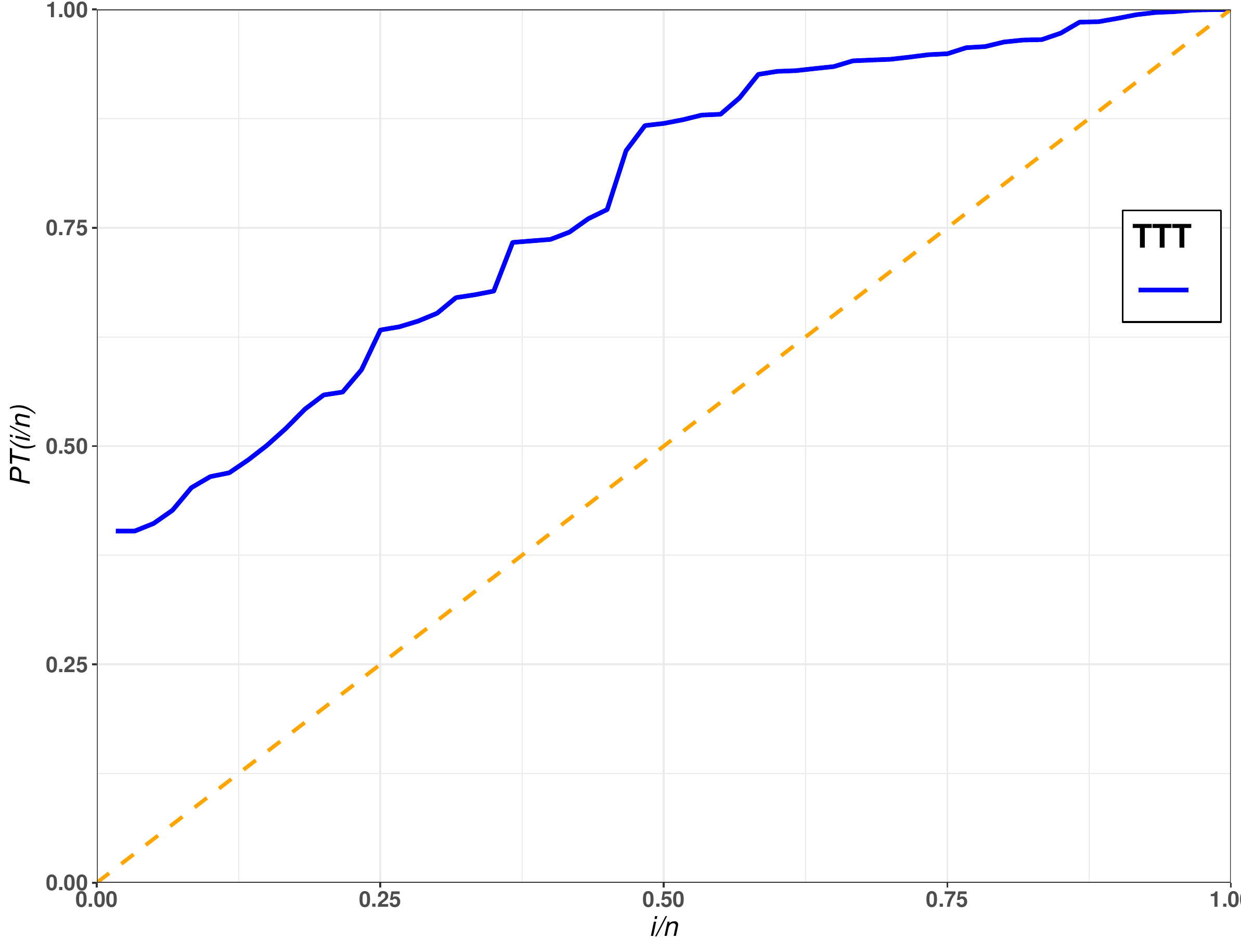}
  \caption{TTT plot}
\end{subfigure}
  \caption{The HRF plot of the GRL distribution and TTT plot for epicenter data}
  \label{fig:fig6}
\end{figure}

\section{Concluding remarks}
\label{sec6}

In this paper, we introduce a new two-parameter distribution called generalized Ramos-Louzada (GRL) distribution. Further, the mathematical properties of the GRL model are studied in detail. The GRL parameters are estimated by eight estimation methods namely: the weighted least-squares, ordinary least squares, maximum likelihood,  maximum product of spacing, Cram\'{e}r--von Mises, Anderson--Darling, Right-tail Anderson--Darling and percentile based estimators. The simulation study illustrates that the maximum product of spacing estimation method outperforms all other estimation methods. Therefore, depends on our study, we can confirm the superiority of the maximum product of spacing method for the GRL distribution. Finally, the practical importance of GRL model was reported in two real applications. The goodness of fit for the proposed data sets showed that our model returned better fitting in comparison with other well-known distributions. Further, the two real data applications show that the maximum product of spacing estimator for leukemia data and least-square estimator for epicenter data return the best estimates for the parameters of the GRL distribution.

\newpage
\appendix
\section*{Appendix A: Tables}
\setcounter{table}{1} \renewcommand{\thetable}{\arabic{table}}

\begin{table}[!h]
\centering
\caption{Simulation results for $\pmb \phi=(\lambda=2.0,\alpha=0.5)^{\intercal}$}%
 \label{tab:tab2}
\resizebox{\textwidth}{!}{
\begin{tabular}{lllllllllll}
\hline \hline \\
$n$ & Est.         & Est. Par.     & WLSE                     & OLSE                       & MLE                        & MPSE                         & CVME                       & ADE                        & RADE                     & PCE                    \\  \\ \hline \\
30  & $|BIAS|$     & $\hat \lambda$ & $ 0.21560 ^{ \{ 4 \} } $    & $ 0.25992 ^{ \{ 7 \} } $ & $ 0.19097 ^{ \{ 2 \} } $   & $ 0.24975 ^{ \{ 6 \} } $ & $ 0.19471 ^{ \{ 3 \} } $ & $ 0.18910 ^{ \{ 1 \} } $   & $ 0.23949 ^{ \{ 5 \} } $   & $ 0.60664 ^{ \{ 8 \} } $ \\
    &              & $\hat \alpha$  & $ 0.04161 ^{ \{ 6 \} } $   & $ 0.04378 ^{ \{ 8 \} } $ & $ 0.03703 ^{ \{ 2 \} } $   & $ 0.03455 ^{ \{ 1 \} } $ & $ 0.04313 ^{ \{ 7 \} } $ & $ 0.03896 ^{ \{ 4 \} } $   & $ 0.03856 ^{ \{ 3 \} } $   & $ 0.03931 ^{ \{ 5 \} } $ \\
    & $MSE$        & $\hat \lambda$ & $ 0.18356 ^{ \{ 3 \} } $   & $ 0.23965 ^{ \{ 5 \} } $ & $ 0.36520 ^{ \{ 7 \} } $   & $ 0.27292 ^{ \{ 6 \} } $ & $ 0.17386 ^{ \{ 2 \} } $ & $ 0.16253 ^{ \{ 1 \} } $   & $ 0.22387 ^{ \{ 4 \} } $   & $ 2.80485 ^{ \{ 8 \} } $ \\
    &              & $\hat \alpha$  & $ 0.00315 ^{ \{ 5 \} } $   & $ 0.00347 ^{ \{ 7 \} } $ & $ 0.00283 ^{ \{ 4 \} } $   & $ 0.00210 ^{ \{ 1 \} } $ & $ 0.00337 ^{ \{ 6 \} } $ & $ 0.00262 ^{ \{ 3 \} } $   & $ 0.00259 ^{ \{ 2 \} } $   & $ 0.00516 ^{ \{ 8 \} } $ \\
\hline \hline
\end{tabular}}
\end{table}

\begin{table}[tbp]
\centering
\caption{Simulation results for $\pmb \phi=(\lambda=2.0,\alpha=2.5)^{\intercal}$}%
 \label{tab:tab3}
\resizebox{\textwidth}{!}{
\begin{tabular}{lllllllllll}
\hline \hline \\
$n$ & Est.         & Est. Par.     & WLSE                     & OLSE                       & MLE                        & MPSE                         & CVME                       & ADE                        & RADE                     & PCE                    \\  \\ \hline \\
30  & $|BIAS|$     & $\hat \lambda$ & $ 0.21244 ^{ \{ 4 \} } $ & $ 0.26189 ^{ \{ 8 \} } $ & $ 0.19438 ^{ \{ 1 \} } $ & $ 0.25088 ^{ \{ 7 \} } $ & $ 0.19742 ^{ \{ 3 \} } $ & $ 0.19522 ^{ \{ 2 \} } $ & $ 0.24708 ^{ \{ 6 \} } $ & $ 0.24484 ^{ \{ 5 \} } $ \\
    &              & $\hat \alpha$  & $ 0.19894 ^{ \{ 6 \} } $ & $ 0.22582 ^{ \{ 8 \} } $ & $ 0.18894 ^{ \{ 3 \} } $ & $ 0.16969 ^{ \{ 1 \} } $ & $ 0.21786 ^{ \{ 7 \} } $ & $ 0.19693 ^{ \{ 5 \} } $ & $ 0.19001 ^{ \{ 4 \} } $ & $ 0.18423 ^{ \{ 2 \} } $ \\
    & $MSE$        & $\hat \lambda$ & $ 0.19082 ^{ \{ 3 \} } $ & $ 0.24426 ^{ \{ 6 \} } $ & $ 0.37099 ^{ \{ 8 \} } $ & $ 0.27879 ^{ \{ 7 \} } $ & $ 0.17947 ^{ \{ 2 \} } $ & $ 0.17346 ^{ \{ 1 \} } $ & $ 0.23687 ^{ \{ 5 \} } $ & $ 0.22556 ^{ \{ 4 \} } $ \\
    &              & $\hat \alpha$  & $ 0.07158 ^{ \{ 5 \} } $ & $ 0.09018 ^{ \{ 8 \} } $ & $ 0.07436 ^{ \{ 6 \} } $ & $ 0.05047 ^{ \{ 1 \} } $ & $ 0.08317 ^{ \{ 7 \} } $ & $ 0.06861 ^{ \{ 4 \} } $ & $ 0.06401 ^{ \{ 3 \} } $ & $ 0.05827 ^{ \{ 2 \} } $ \\
\hline \hline
\end{tabular}}
\end{table}

\begin{table}[tbp]
\centering
\caption{Simulation results for $\pmb \phi=(\lambda=2.0,\alpha=0.7)^{\intercal}$}%
 \label{tab:tab4}
\resizebox{\textwidth}{!}{
\begin{tabular}{lllllllllll}
\hline \hline \\
$n$ & Est.         & Est. Par.      & WLSE                     & OLSE                     & MLE                      & MPSE                        & CVME                     & ADE                      & RADE                     & PCE                    \\  \\ \hline \\
30  & $|BIAS|$     & $\hat \lambda$ & $ 0.21445 ^{ \{ 4 \} } $   & $ 0.25938 ^{ \{ 7 \} } $ & $ 0.19354 ^{ \{ 2 \} } $ & $ 0.24781 ^{ \{ 6 \} } $ & $ 0.19844 ^{ \{ 3 \} } $ & $ 0.18907 ^{ \{ 1 \} } $ & $ 0.24632 ^{ \{ 5 \} } $   & $ 0.34260 ^{ \{ 8 \} } $ \\
    &              & $\hat \alpha$  & $ 0.05581 ^{ \{ 3 \} } $   & $ 0.08771 ^{ \{ 8 \} } $ & $ 0.07506 ^{ \{ 5 \} } $ & $ 0.04727 ^{ \{ 1 \} } $ & $ 0.08616 ^{ \{ 7 \} } $ & $ 0.07770 ^{ \{ 6 \} } $ & $ 0.05562 ^{ \{ 2 \} } $   & $ 0.07092 ^{ \{ 4 \} } $ \\
    & $MSE$        & $\hat \lambda$ & $ 0.18452 ^{ \{ 3 \} } $   & $ 0.23540 ^{ \{ 4 \} } $ & $ 0.36081 ^{ \{ 7 \} } $ & $ 0.27295 ^{ \{ 6 \} } $ & $ 0.18169 ^{ \{ 2 \} } $ & $ 0.16400 ^{ \{ 1 \} } $ & $ 0.23850 ^{ \{ 5 \} } $   & $ 0.37356 ^{ \{ 8 \} } $ \\
    &              & $\hat \alpha$  & $ 0.00563 ^{ \{ 3 \} } $   & $ 0.01365 ^{ \{ 8 \} } $ & $ 0.01143 ^{ \{ 6 \} } $ & $ 0.00388 ^{ \{ 1 \} } $ & $ 0.01301 ^{ \{ 7 \} } $ & $ 0.01055 ^{ \{ 5 \} } $ & $ 0.00535 ^{ \{ 2 \} } $   & $ 0.00939 ^{ \{ 4 \} } $ \\
\hline \hline
\end{tabular}}
\end{table}

\begin{table}[tbp]
\centering
\caption{Simulation results for $\pmb \phi=(\lambda=2.0,\alpha=3.5)^{\intercal}$}%
 \label{tab:tab5}
\resizebox{\textwidth}{!}{
\begin{tabular}{lllllllllll}
\hline \hline \\
$n$ & Est.         & Est. Par.      & WLSE                     & OLSE                     & MLE                      & MPSE                        & CVME                     & ADE                      & RADE                     & PCE                    \\  \\ \hline \\
30  & $|BIAS|$     & $\hat \lambda$ & $ 0.22261 ^{ \{ 4 \} } $ & $ 0.24906 ^{ \{ 7 \} } $ & $ 0.20385 ^{ \{ 3 \} } $ & $ 0.25031 ^{ \{ 8 \} } $ & $ 0.19876 ^{ \{ 2 \} } $ & $ 0.19204 ^{ \{ 1 \} } $ & $ 0.23014 ^{ \{ 6 \} } $ & $ 0.22406 ^{ \{ 5 \} } $ \\
    &              & $\hat \alpha$  & $ 0.08131 ^{ \{ 1 \} } $ & $ 0.30396 ^{ \{ 7 \} } $ & $ 0.26007 ^{ \{ 4 \} } $ & $ 0.24244 ^{ \{ 2 \} } $ & $ 0.30685 ^{ \{ 8 \} } $ & $ 0.27002 ^{ \{ 5 \} } $ & $ 0.27073 ^{ \{ 6 \} } $ & $ 0.25486 ^{ \{ 3 \} } $ \\
    & $MSE$        & $\hat \lambda$ & $ 0.20562 ^{ \{ 4 \} } $ & $ 0.22146 ^{ \{ 6 \} } $ & $ 0.40083 ^{ \{ 8 \} } $ & $ 0.27667 ^{ \{ 7 \} } $ & $ 0.18213 ^{ \{ 2 \} } $ & $ 0.17251 ^{ \{ 1 \} } $ & $ 0.21507 ^{ \{ 5 \} } $ & $ 0.19057 ^{ \{ 3 \} } $ \\
    &              & $\hat \alpha$  & $ 0.01216 ^{ \{ 1 \} } $ & $ 0.16235 ^{ \{ 7 \} } $ & $ 0.14066 ^{ \{ 6 \} } $ & $ 0.10401 ^{ \{ 2 \} } $ & $ 0.16608 ^{ \{ 8 \} } $ & $ 0.12814 ^{ \{ 4 \} } $ & $ 0.12825 ^{ \{ 5 \} } $ & $ 0.11077 ^{ \{ 3 \} } $ \\
\hline \hline
\end{tabular}}
\end{table}

\begin{table}[tbp]
\centering
\caption{Simulation results for $\pmb \phi=(\lambda=3.1,\alpha=0.5)^{\intercal}$}%
 \label{tab:tab6}
\resizebox{\textwidth}{!}{
\begin{tabular}{lllllllllll}
\hline \hline \\
$n$ & Est.         & Est. Par.      & WLSE                     & OLSE                     & MLE                      & MPSE                        & CVME                     & ADE                      & RADE                     & PCE                    \\  \\ \hline \\
30  & $|BIAS|$     & $\hat \lambda$ & $ 0.71536 ^{ \{ 6 \} } $ & $ 0.42777 ^{ \{ 2 \} } $ & $ 1.54000 ^{ \{ 7 \} } $ & $ 0.34691 ^{ \{ 1 \} } $ & $ 0.49932 ^{ \{ 4 \} } $ & $ 0.47997 ^{ \{ 3 \} } $ & $ 0.64058 ^{ \{ 5 \} } $ & $ 3.17604 ^{ \{ 8 \} } $  \\
    &              & $\hat \alpha$  & $ 0.05196 ^{ \{ 4 \} } $ & $ 0.05374 ^{ \{ 5 \} } $ & $ 0.06587 ^{ \{ 7 \} } $ & $ 0.03983 ^{ \{ 1 \} } $ & $ 0.05536 ^{ \{ 6 \} } $ & $ 0.04617 ^{ \{ 3 \} } $ & $ 0.04508 ^{ \{ 2 \} } $ & $ 0.10036 ^{ \{ 8 \} } $  \\
    & $MSE$        & $\hat \lambda$ & $ 5.59864 ^{ \{ 6 \} } $ & $ 0.53566 ^{ \{ 4 \} } $ & $ 7.76717 ^{ \{ 7 \} } $ & $ 0.45223 ^{ \{ 3 \} } $ & $ 0.44662 ^{ \{ 2 \} } $ & $ 0.42723 ^{ \{ 1 \} } $ & $ 2.97747 ^{ \{ 5 \} } $ & $ 40.41914 ^{ \{ 8 \} } $ \\
    &              & $\hat \alpha$  & $ 0.00488 ^{ \{ 5 \} } $ & $ 0.00471 ^{ \{ 4 \} } $ & $ 0.00846 ^{ \{ 7 \} } $ & $ 0.00246 ^{ \{ 1 \} } $ & $ 0.00507 ^{ \{ 6 \} } $ & $ 0.00325 ^{ \{ 2 \} } $ & $ 0.00328 ^{ \{ 3 \} } $ & $ 0.02050 ^{ \{ 8 \} } $   \\
\hline \hline
\end{tabular}}
\end{table}

\begin{table}[tbp]
\centering
\caption{Simulation results for $\pmb \phi=(\lambda=3.1,\alpha=2.5)^{\intercal}$}%
 \label{tab:tab7}
\resizebox{\textwidth}{!}{
\begin{tabular}{lllllllllll}
\hline \hline \\
$n$ & Est.         & Est. Par.      & WLSE                     & OLSE                     & MLE                      & MPSE                        & CVME                     & ADE                      & RADE                     & PCE                    \\  \\ \hline \\
30  & $|BIAS|$     & $\hat \lambda$ & $ 0.51721 ^{ \{ 2 \} } $ & $ 0.52263 ^{ \{ 3 \} } $ & $ 1.49692 ^{ \{ 8 \} } $ & $ 0.38546 ^{ \{ 1 \} } $ & $ 0.69528 ^{ \{ 4 \} } $ & $ 0.79067 ^{ \{ 5 \} } $ & $ 1.09366 ^{ \{ 7 \} } $  & $ 0.79091 ^{ \{ 6 \} } $ \\
    &              & $\hat \alpha$  & $ 0.24577 ^{ \{ 5 \} } $ & $ 0.25919 ^{ \{ 6 \} } $ & $ 0.30569 ^{ \{ 8 \} } $ & $ 0.19785 ^{ \{ 1 \} } $ & $ 0.27143 ^{ \{ 7 \} } $ & $ 0.22402 ^{ \{ 3 \} } $ & $ 0.23879 ^{ \{ 4 \} } $  & $ 0.22054 ^{ \{ 2 \} } $ \\
    & $MSE$        & $\hat \lambda$ & $ 0.75318 ^{ \{ 2 \} } $ & $ 1.88313 ^{ \{ 3 \} } $ & $ 7.25731 ^{ \{ 7 \} } $ & $ 0.57621 ^{ \{ 1 \} } $ & $ 4.65183 ^{ \{ 5 \} } $ & $ 3.60637 ^{ \{ 4 \} } $ & $ 10.39735 ^{ \{ 8 \} } $ & $ 4.93878 ^{ \{ 6 \} } $ \\
    &              & $\hat \alpha$  & $ 0.10220 ^{ \{ 4 \} } $ & $ 0.11081 ^{ \{ 5 \} } $ & $ 0.17793 ^{ \{ 8 \} } $ & $ 0.05913 ^{ \{ 1 \} } $ & $ 0.12940 ^{ \{ 7 \} } $ & $ 0.08607 ^{ \{ 3 \} } $ & $ 0.11464 ^{ \{ 6 \} } $  & $ 0.08595 ^{ \{ 2 \} } $ \\
\hline \hline
\end{tabular}}
\end{table}

\begin{table}[tbp]
\centering
\caption{Simulation results for $\pmb \phi=(\lambda=3.1,\alpha=0.7)^{\intercal}$}%
 \label{tab:tab8}
\resizebox{\textwidth}{!}{
\begin{tabular}{lllllllllll}
\hline \hline \\
$n$ & Est.         & Est. Par.      & WLSE                     & OLSE                     & MLE                      & MPSE                        & CVME                     & ADE                      & RADE                     & PCE                    \\  \\ \hline \\
30  & $|BIAS|$     & $\hat \lambda$ & $ 0.51044 ^{ \{ 4 \} } $ & $ 0.41892 ^{ \{ 2 \} } $ & $ 1.49321 ^{ \{ 8 \} } $ & $ 0.35271 ^{ \{ 1 \} } $ & $ 0.49185 ^{ \{ 3 \} } $ & $ 0.73176 ^{ \{ 6 \} } $ & $ 0.56808 ^{ \{ 5 \} } $ & $ 1.25874 ^{ \{ 7 \} } $ \\
    &              & $\hat \alpha$  & $ 0.07146 ^{ \{ 3 \} } $ & $ 0.10634 ^{ \{ 5 \} } $ & $ 0.12671 ^{ \{ 8 \} } $ & $ 0.05599 ^{ \{ 1 \} } $ & $ 0.10847 ^{ \{ 6 \} } $ & $ 0.08794 ^{ \{ 4 \} } $ & $ 0.06330 ^{ \{ 2 \} } $ & $ 0.12407 ^{ \{ 7 \} } $ \\
    & $MSE$        & $\hat \lambda$ & $ 0.64958 ^{ \{ 4 \} } $ & $ 0.35564 ^{ \{ 1 \} } $ & $ 6.96203 ^{ \{ 7 \} } $ & $ 0.44103 ^{ \{ 3 \} } $ & $ 0.43581 ^{ \{ 2 \} } $ & $ 3.25229 ^{ \{ 6 \} } $ & $ 1.32980 ^{ \{ 5 \} } $ & $ 7.09123 ^{ \{ 8 \} } $ \\
    &              & $\hat \alpha$  & $ 0.00833 ^{ \{ 3 \} } $ & $ 0.01841 ^{ \{ 5 \} } $ & $ 0.03025 ^{ \{ 8 \} } $ & $ 0.00473 ^{ \{ 1 \} } $ & $ 0.01896 ^{ \{ 6 \} } $ & $ 0.01283 ^{ \{ 4 \} } $ & $ 0.00615 ^{ \{ 2 \} } $ & $ 0.03006 ^{ \{ 7 \} } $ \\
\hline \hline
\end{tabular}}
\end{table}

\begin{table}[tbp]
\centering
\caption{Simulation results for $\pmb \phi=(\lambda=3.1,\alpha=3.5)^{\intercal}$}%
 \label{tab:tab9}
\resizebox{\textwidth}{!}{
\begin{tabular}{lllllllllll}
\hline \hline \\
$n$ & Est.         & Est. Par.      & WLSE                     & OLSE                     & MLE                      & MPSE                        & CVME                     & ADE                      & RADE                     & PCE                    \\  \\ \hline     \\
30  & $|BIAS|$     & $\hat \lambda$ & $ 0.52557 ^{ \{ 2 \} } $ & $ 1.09506 ^{ \{ 4 \} } $ & $ 1.37011 ^{ \{ 8 \} } $ & $ 0.38222 ^{ \{ 1 \} } $ & $ 1.32080 ^{ \{ 7 \} } $  & $ 1.12191 ^{ \{ 5 \} } $ & $ 1.26170 ^{ \{ 6 \} } $  & $ 0.90329 ^{ \{ 3 \} } $ \\
    &              & $\hat \alpha$  & $ 0.09999 ^{ \{ 1 \} } $ & $ 0.39092 ^{ \{ 6 \} } $ & $ 0.40091 ^{ \{ 7 \} } $ & $ 0.26417 ^{ \{ 2 \} } $ & $ 0.41516 ^{ \{ 8 \} } $  & $ 0.34797 ^{ \{ 5 \} } $ & $ 0.34765 ^{ \{ 4 \} } $  & $ 0.31293 ^{ \{ 3 \} } $ \\
    & $MSE$        & $\hat \lambda$ & $ 0.73033 ^{ \{ 2 \} } $ & $ 9.10142 ^{ \{ 6 \} } $ & $ 5.90223 ^{ \{ 4 \} } $ & $ 0.64502 ^{ \{ 1 \} } $ & $ 13.03684 ^{ \{ 8 \} } $ & $ 6.92476 ^{ \{ 5 \} } $ & $ 10.48503 ^{ \{ 7 \} } $ & $ 4.13403 ^{ \{ 3 \} } $ \\
    &              & $\hat \alpha$  & $ 0.01745 ^{ \{ 1 \} } $ & $ 0.32306 ^{ \{ 7 \} } $ & $ 0.29950 ^{ \{ 6 \} } $ & $ 0.10485 ^{ \{ 2 \} } $ & $ 0.37929 ^{ \{ 8 \} } $  & $ 0.23413 ^{ \{ 4 \} } $ & $ 0.25531 ^{ \{ 5 \} } $  & $ 0.18394 ^{ \{ 3 \} } $ \\
\hline \hline
\end{tabular}}
\end{table}

\end{document}